\documentclass[a4paper,USenglish]{lipics-v2021}
\hideLIPIcs

\usepackage{balance}
\usepackage{algorithm}
\usepackage[noend]{algpseudocode}
\usepackage{ifthen}
\usepackage{comment}
\usepackage{amsthm,amsmath}
\usepackage{graphicx}

\usepackage{pstricks,pst-node,pst-tree}
\usepackage{soul}

\usepackage{color}
\usepackage{mathrsfs}

\renewcommand{\epsilon}{\varepsilon}

\newcommand{\Prob}[1]{\hbox{\rm I\kern-2pt P}\left[#1\right]}

\nolinenumbers

\DeclareMathAlphabet{\mathsc}{OT1}{cmr}{m}{sc}

\newcommand{\cH}{{\mathcal{H}}}

\def\NumOfBaseFragments{\textsf{NUM-OF-BASE-FRAGMENTS}}
\def\ClusterEdges{\textsf{CLUSTER-EDGES}}
\def\ID{\textsf{ID}}
\def\ClusterID{\textsf{CLUSTER-ID}}

\def\MOEValue{\textsf{MOE-VALUE}}
\def\FinalMOEValue{\textsf{FINAL-MOE-VALUE}}

\def\IsFragmentRoot{\textsf{IS-FRAGMENT-ROOT}}

\def\GraphLeader{\mathcal{L}}
\def\BFSTree{\mathcal{T}}

\def\MSTConstruction{\textbf{Sing-MST}}
\def\LeaderElection{\textbf{LE}}
\def\TreeCounting{\textbf{Tree-Count}}

\def\FindMOE{\textbf{Find-MOE}}

\def\DiameterCalculation{\textbf{Diam-Calc}}
\def\ControlledGHS{\textbf{Controlled-GHS}}

\def\ColeVishkin{\textbf{Cole-Vishkin}}

\def\SpanningTreeConstruction{\textbf{ST-Cons}}
\def\Upcast{\textbf{Upcast}}
\def\Downcast{\textbf{Downcast}}
\def\FragmentBroadcast{\textbf{Frag-Bcast}}
\def\MPX{\textbf{MPX}}
\def\Transform{\textbf{Transform}}
\def\imax{i_{m}}
\def\intercluster{inter-cluster}
\def\CLT{\hat{T}}
\def\SPT{\tilde{T}}
\def\BFST{T^{BFS}}
\def\SUPT{T^{super}}
\def\diam{diam}

\def\Vstart[#1]{V^{start}_{#1}}

\def\Vactive[#1]{V^{act}_{#1}}

\def\EV{\mathcal{E}}

\def\polylog{\operatorname{polylog}}

\newcommand{\squishlist}{
 \begin{list}{$\bullet$}
  { \setlength{\itemsep}{0pt}
     \setlength{\parsep}{3pt}
     \setlength{\topsep}{3pt}
     \setlength{\partopsep}{0pt}
     \setlength{\leftmargin}{1.5em}
     \setlength{\labelwidth}{1em}
     \setlength{\labelsep}{0.5em} } }

\newcommand{\squishlisttwo}{
 \begin{list}{$\bullet$}
  { \setlength{\itemsep}{0pt}
     \setlength{\parsep}{0pt}
    \setlength{\topsep}{0pt}
    \setlength{\partopsep}{0pt}
    \setlength{\leftmargin}{2em}
    \setlength{\labelwidth}{1.5em}
    \setlength{\labelsep}{0.5em} } }

\newcommand{\squishend}{
  \end{list}  }

\title{An Almost Singularly Optimal Asynchronous Distributed MST Algorithm}

\author{Fabien Dufoulon}{Department of Computer Science, University of Houston, Houston, TX, USA}{fabien.dufoulon.cs@gmail.com}{https://orcid.org/0000-0003-2977-4109}{This work was supported in part by NSF grants CCF-1717075, CCF-1540512, IIS-1633720, and BSF grant 2016419.}

\author{Shay Kutten}{Faculty of Industrial Engineering and Management, Technion - Israel Institute of Technology, Haifa, Israel}{kutten@technion.ac.il}{https://orcid.org/0000-0003-2062-6855}{This work was supported in part by the Bi-national Science Foundation (BSF) grant 2016419 and supported in part by ISF grant 1346/22.}

\author{William K. {Moses Jr.}}{Department of Computer Science, University of Houston, Houston, TX, USA}{wkmjr3@gmail.com}{https://orcid.org/0000-0002-4533-7593}{This work was supported in part by NSF grants CCF1540512, IIS-1633720, CCF-1717075, and BSF grant 2016419.}

\author{Gopal Pandurangan}{Department of Computer Science, University of Houston, Houston, TX, USA}{gopal@cs.uh.edu}{https://orcid.org/0000-0001-5833-6592}{This work was supported in part by NSF grants CCF-1717075, CCF-1540512, IIS-1633720, and BSF grant 2016419.}

\author{David Peleg}{Department of Computer Science and Applied Mathematics, Weizmann Institute of Science, Rehovot, Israel}{david.peleg@weizmann.ac.il}{https://orcid.org/0000-0003-1590-0506}{This work was supported in part by the US-Israel Binational Science Foundation grant 2018043.}

\authorrunning{F. Dufoulon, S. Kutten, W.\,K. Moses Jr., G. Pandurangan, and D. Peleg}
\Copyright{Fabien Dufoulon, Shay Kutten, William K. Moses Jr., Gopal Pandurangan, and David Peleg}



\EventEditors{Christian Scheideler}
\EventNoEds{1}
\EventLongTitle{36th International Symposium on Distributed Computing (DISC 2022)}
\EventShortTitle{DISC 2022}
\EventAcronym{DISC}
\EventYear{2022}
\EventDate{October 25--27, 2022}
\EventLocation{Augusta, Georgia, USA}
\EventLogo{}
\SeriesVolume{246}
\ArticleNo{17}

\begin{document}
\maketitle

\begin{abstract}
A singularly (near) optimal distributed algorithm is one that is (near) optimal in \emph{two} criteria, namely,
its time and message complexities. For \emph{synchronous} $\mathcal{CONGEST}$ networks, such algorithms are known for fundamental
distributed computing problems such
as leader election [Kutten et al., JACM 2015] and
Minimum Spanning Tree
(MST) construction [Pandurangan et al., STOC 2017, Elkin, PODC 2017].  
However, it is  open whether a singularly (near) optimal bound can be obtained  for the MST construction problem in general \emph{asynchronous} $\mathcal{CONGEST}$ networks.

In this paper, we present 
 a randomized distributed MST  algorithm  that, with high probability, computes an MST in \emph{asynchronous} $\mathcal{CONGEST}$ networks  and  takes $\Tilde{O}(D^{1+\epsilon} + \sqrt{n})$ time and $\Tilde{O}(m)$ messages\footnote{The $\Tilde{O}$ notation hides a $\polylog(n)$ factor and the $\Tilde{\Omega}$ notation hides a $1/\polylog(n)$ factor.},
 where $n$ is the number of  nodes, $m$ the number of edges,  $D$ is the diameter of the network, and $\epsilon >0$ is an arbitrarily  small constant (both time and message bounds hold with high probability). Since $\tilde{\Omega}(D+\sqrt{n})$ and $\Omega(m)$ are respective time and message lower bounds for
 distributed MST construction in the standard $KT_0$ model, our algorithm  is message optimal (up to a $\polylog(n)$ factor) and almost time optimal (except for a $D^{\epsilon}$ factor).
 Our result   answers an open question raised in Mashregi and King [DISC 2019] by giving the first known  asynchronous MST algorithm that has sublinear time (for all $D = O(n^{1-\epsilon})$) and uses $\tilde{O}(m)$ messages. Using a result of Mashregi and King [DISC 2019], this also yields the first asynchronous MST algorithm that is sublinear in both time and messages in the
 $KT_1$ $\mathcal{CONGEST}$ model.

 A key tool in our algorithm is the construction of a low diameter rooted spanning tree in asynchronous $\mathcal{CONGEST}$ 
 that has depth $\tilde{O}(D^{1+\epsilon})$ (for an arbitrarily small constant $\epsilon > 0$) in $\tilde{O}(D^{1+\epsilon})$ time 
and $\tilde{O}(m)$ messages.
 To the best of our knowledge, this is the first such construction that is almost singularly optimal in the asynchronous setting. This tree construction 
 may be 
 of independent interest as it can also be used for efficiently performing basic  tasks such as {\em verified}
broadcast and convergecast  in asynchronous networks.

\end{abstract}

\section{Introduction}
\label{sec:intro}

\subsection{Background and Motivation}

Singularly (near) optimal distributed algorithms are those that are (near) optimal both in their message complexity and in their time complexity.\footnote{In this paper,
henceforth, when we say ``near optimal'' we  mean ``optimal up to a $\polylog(n)$ factor'', where $n$ is the network size.} 
The current paper is intended as a  step in expanding the study of ``which problems admit singularly optimal algorithms'' from the realm of synchronous $\mathcal{CONGEST}$ networks to that of \emph{asynchronous} ones.

An important example of a problem that has been studied in the context of singularly (near) optimal algorithms is 
minimum-weight spanning tree (MST)
construction. 
This has become a rather canonical problem in the sub area of distributed graph algorithms and was used to demonstrate and study various concepts such as the congested clique model (Lotker et al.~\cite{LotkerPPP05}), proof labeling schemes (Korman et al.~\cite{KormanKP05}), networks with latency and capacity (Augustine et al.~\cite{augustine2020latency}), cognitive radio networks (Rohilla et al.~\cite{rohilla2020efficient}),  
distributed applications of graph sketches (King et al.~\cite{KingKT15}), 
distributed computing with advice (Fraigniaud et al.~\cite{fraigniaud2010local}), 
distributed verification and hardness of approximation (Kor et al.~\cite{KorKP11}, Korman and Kutten~\cite{korman2007distributed} and
Das Sarma et al.~\cite{stoc11}),
self-stabilizing algorithms (Gupta and Srimani~\cite{gupta2003self} and many other papers), distributed quantum computing
(Elkin et al.~\cite{quantum}) and more.
The study of the MST problem in what we now call the $\mathcal{CONGEST}$ model
started more than forty years ago, see Dalal, and also Spira~\cite{dalal-mst1,dalal-mst,spira1977communication}. 

The seminal paper of Gallager, Humblet, and Spira (GHS)~\cite{GallagerHS83}  presented a
distributed algorithm for an \emph{asynchronous} network that constructs an MST in $O(n \log n)$ time using $O(m + n \log n)$ messages,
where $n$ and $m$ denote the number of nodes and the number of edges of the network, respectively.
The time complexity was later improved by 
 Awerbuch and by Faloutsos and Moelle to $O(n)$~\cite{awerbuch-optimal,faloutsos}, while keeping the same order of message complexity.

The message complexity of GHS  algorithm is (essentially) optimal, since it can be shown that for any 
$1 \leq m \leq n^2$, there exists a graph with $\Theta(m)$ edges such that $\Omega(m)$
is a lower bound on the message complexity of constructing even a spanning tree (even for randomized algorithms)~\cite{jacm15}.\footnote{This message lower bound holds in the so-called $KT_0$ model, which is assumed in this paper. See Section \ref{subsec:related-work} for more details.}
Moreover, the time complexity bound of $O(n)$ bound is {\em existentially} optimal (in the sense that there exist graphs (of high diameter) for which this is the best possible).
However, the time bound is not optimal if one parameterizes the running time in terms of the network diameter $D$, which can be much smaller than $n$.
In a \emph{synchronous} network, Garay, Kutten, and Peleg~\cite{garay-sublinear} gave the
first such distributed algorithm for the MST problem with running
time $\tilde{O}(D + n^{0.614})$, which was later improved by Kutten and
Peleg~\cite{kutten-domset} to $\tilde{O}(D + \sqrt{n})$ (again for a synchronous network).
However, both these algorithms
are not message-optimal
as they exchange $O(m + n^{1.614})$ and $O(m + n^{1.5})$ messages, respectively. 

Conversely, it was established by Peleg and Rubinovich~\cite{peleg-bound} that $\tilde{\Omega}(D+ \sqrt{n})$ is a lower
bound 
on the time complexity of distributed MST construction that applies even to low-diameter networks ($D = \Omega(\log n)$), and to the synchronous setting.  
The lower bound of Peleg and Rubinovich applies to exact, deterministic algorithms.
This lower bound was further extended to randomized (Monte Carlo) algorithms, approximate constructions, MST verification, and more (see~\cite{LotkerPP01,
LotkerPPP05,elkin,stoc11}).

Pandurangan, Robinson and Scquizzato 
\cite{PanduranganRS17,pandurangan2019time}  
showed that MST admits a randomized singularly near optimal algorithm
in \emph{synchronous} $\mathcal{CONGEST}$ networks; their algorithm  uses $\tilde{O}(m)$ messages and $\tilde{O}(D+\sqrt{n})$ rounds. Subsequently, Elkin~\cite{Elkin17} presented
a simpler, singularly optimal deterministic MST algorithm, again for synchronous networks.

For \emph{asynchronous} networks,  one can obtain algorithms that are separately time optimal (by combining~\cite{kutten-domset} with a synchronizer, see Awerbuch~\cite{awerbuch1985complexity}) or message optimal~\cite{GallagerHS83} for the MST problem, but it is open  whether one can obtain an asynchronous distributed MST algorithm that is
singularly (near) optimal. 
This is one of the main motivations for this work. 
An additional motivation 
is to design tools that can be useful for  constructing singularly  optimal algorithms for other fundamental problems in asynchronous networks.

In general, designing singularly optimal algorithms for \emph{asynchronous} networks seems harder compared to synchronous networks.
In \emph{synchronous} networks, besides MST construction, singularly (near) optimal algorithms have been shown in recent years for leader election,
 (approximate) shortest paths, and several other 
 problems~\cite{jacm15,haeupler2018round}. However, all these results  \emph{do not} apply to asynchronous networks. 
 Converting synchronous algorithms to work on asynchronous networks generally 
incur heavy cost overhead, increasing either time or message complexity or both substantially. In particular, using \emph{synchronizers} 
\cite{awerbuch1985complexity} to convert a singularly optimal algorithm to work in an asynchronous network generally renders
the asynchronous algorithm not singularly optimal.  
Using a synchronizer can significantly increase either the time or the message complexity or both far beyond the complexities of the algorithm presented here. Furthermore, there can be a non-trivial
cost associated with \emph{constructing}
such a synchronizer in the first place.

For example, applying the simple $\alpha$ synchronizer~\cite{awerbuch1985complexity} (which does not require the a priori existence of a leader or a spanning tree)
to the singularly optimal synchronous MST algorithm of~\cite{PanduranganRS17,pandurangan2019time} or~\cite{Elkin17} yields an asynchronous algorithm
with message complexity of $\tilde{O}(m(D+\sqrt{n}))$ and  time complexity of $\tilde{O}(D+\sqrt{n})$; 
this algorithm is time optimal, but \emph{not} message optimal.  
Some other synchronizers (see, e.g., Awerbuch and Peleg~\cite{AP90}), do construct efficient synchronizers that can achieve near optimal conversion from 
synchronous to asynchronous algorithms with respect to both time and messages, but  constructing the synchronizer itself 
requires a substantial preprocessing or initialization  cost.
For example, the message cost of the synchronizer setup protocol of~\cite{AP90} can be as high as $O(mn)$. 

Another rather tempting idea to derive an MST algorithm that would be efficient both in time and in messages would be to convert a result of Mashreghi and King~\cite{KMDISC19} (see also~\cite{KMDISC18} and discussion in Section \ref{subsec:related-work}), originally designed in the asynchronous $KT_1$ $\mathcal{CONGEST}$ model\footnote{In $KT_1$ model it is assumed that  nodes know the identities of their neighbors (cf. Section \ref{subsec:related-work}), unlike the $KT_0$ model, where nodes don't have that knowledge.} to the 
more common $KT_0$ model assumed here. In particular, they give an asynchronous MST algorithm that takes
$O(n)$ time and $\tilde{O}(n^{1.5})$ messages. Note that one can convert an algorithm in the $KT_1$ model to work in the $KT_0$ model by allowing each node to communicate with
all its neighbors in one round; this takes an additional $\tilde{O}(m)$ messages. 
Hence, with such a conversion  the message complexity of the above algorithm would be essentially optimal (i.e., $\tilde{O}(m)$), but the time complexity would be $O(n)$ which is only existentially optimal, and can be significantly higher than
the lower bound of $\tilde{O}(D+\sqrt{n})$. In fact, as we will discuss later, our result answers an open question
posed in~\cite{KMDISC19} and gives  MST algorithms with improved bounds in asynchronous $KT_1$ model (cf. Section
\ref{subsec:our-contributions}).

Instead of using a synchronizer, a better approach might be to design an algorithm directly for an asynchronous network.
As an example, consider the fundamental leader election problem, which is  simpler than the MST construction
problem. 
Till recently, a singularly optimal asynchronous leader election algorithm was not known. Applying a synchronizer to  known
\emph{synchronous} singularly optimal leader election algorithms {\em does not} yield singularly optimal asynchronous algorithms. For example, applying the simple $\alpha$ synchronizer 
to the singularly optimal synchronous leader election algorithm of~\cite{jacm15} yields an asynchronous algorithm with message complexity of $O(mD\log n)$
and  time complexity of $O(D)$; this algorithm is not message optimal, especially for large diameter networks.
Other synchronizers such as $\beta$ and $\gamma$ of~\cite{awerbuch1985complexity} and that of~\cite{AP90}, require the a priori existence of a \emph{leader} or \emph{a spanning  tree} and 
hence  cannot be used for leader election. 
The work of Kutten et al.~\cite{KMPP21} presented a singularly (near) optimal leader election
 for asynchronous networks that 
  takes  $\Tilde{O}(m)$ messages and $\Tilde{O}(D)$ time.\footnote{This algorithm is singularly near  optimal, since 
  $\Omega(m)$ and $\Omega(D)$ are message and lower bounds for leader election even for randomized Monte Carlo algorithms~\cite{jacm15}.}
  That algorithm did not use a synchronizer and was directly designed for an asynchronous network. The leader election algorithm
  of~\cite{KMPP21} is a useful subroutine in our MST algorithm.

\subsection{The Distributed Computing Model}
\label{subsec:model}

The distributed network is modeled  as an arbitrary undirected connected weighted graph $G=(V,E,w)$, 
where the node set $V$ represent the processors, the edge set $E$ represents the communication
links between them, and $w(e)$ is the weight of edge $e \in E$.  $D$ denotes the hop-diameter (that is, the unweighted
diameter) of $G$,  in this paper, diameter always means hop-diameter.
We also assume that the weights of the edges of
the graph are all distinct. This implies that the MST of the graph is unique.
(The definitions and the results generalize readily to the case where the weights are not necessarily distinct.)
We make the common
assumption that each node has a unique identity (this is not essential, but
simplifies presentation), and at the beginning of computation, each
node $v$ accepts as input its own identity number (ID) and the weights
of the edges incident to it. 
Thus, a node has only {\em local}
knowledge. We assume that each node has ports (each port having a unique port number); each incident edge is connected
to one distinct port. A node \emph{does not} have any initial knowledge of the other endpoint  of its incident edge (the identity of the node it is connected to or the port number that it is connected to). This model is  referred to as the {\em clean network model} in~\cite{peleg-locality} and is also sometimes referred to as the
$KT_0$ model, i.e., the initial (K)nowledge of
all nodes is restricted (T)ill radius 0 (i.e., just the local knowledge)~\cite{peleg-locality}.
The $KT_0$ model is extensively used in distributed computing literature including MST algorithms
(see e.g.,~\cite{peleg-locality,PanduranganRS18} and the references therein). While we design an algorithm
for the $KT_0$ model, our algorithm also yields an improvement in the $KT_1$ model~\cite{AGVP90,peleg-locality} where each node has an initial knowledge of the identities of its neighbors. 

We assume that nodes have knowledge of $n$ (in fact a constant factor approximation of $n$ is sufficient), the network size.  We note that quite a few prior distributed algorithms require knowledge of
$n$, see e.g.~\cite{A89,SS94,AM94,KMPP21}. 
We assume that  processors can access \emph{private unbiased random bits}. 

We assume the standard \emph{asynchronous} $\mathcal{CONGEST}$ communication
model~\cite{peleg-locality},
where messages (each message is of $O(\log n)$ bits) sent over an edge incur unpredictable 
but finite delays, in an error-free and FIFO manner (i.e., messages will arrive in sequence).
As is standard, it is assumed that a message 
takes \emph{at most one time unit} to be delivered across an edge. Note that this is just for the sake of the analysis of time complexity,
and does not imply that nodes know an upper bound on the delay of any message.
As usual, local computation within a node is assumed to be instantaneous and free; however, our algorithm will involve  only lightweight local computations. 

We assume an \emph{adversarial wake-up} model, where node wake-up times are scheduled by an adversary (who may decide to keep some nodes dormant)
which is standard in prior asynchronous protocols (see~\cite{AG91,GallagerHS83,singh97}). Nodes are initially asleep, and a node enters the execution when it is woken up by 
the environment 
or upon receiving messages from other nodes.\footnote{Although standard, the  adversarial wake up model, in our setting,  is not more difficult
compared to the alternative \emph{simultaneous wake up} model where all nodes are assumed to be awake at the beginning of the computation. Indeed, in the adversarial wake up model,  awake nodes can broadcast (by simply flooding)
a ``wake up'' message which can wake up all nodes; this takes  only $O(m)$ messages and $O(D)$ time and hence within the singularly optimal bounds.}

The time complexity is measured from the moment the first node wakes up.
The adversary  wakes up nodes and  
 delays each message in an \emph{adaptive} fashion, i.e., when the adversary makes a decision to wake up a node or delay a message, 
 it has access to the results of all previous coin flips. 
In the asynchronous setting, once a node enters execution, it performs all the computations required of it by the algorithm, 
and sends out messages to neighbors as specified by the algorithm.
At the end of the computation, we require each node to know which of its incident edges belong to the MST. When we say that an algorithm has termination detection, we mean that all nodes detect termination, i.e., each node detects that its own participation in the algorithm is over.

\subsection{Our Contributions}
\label{subsec:our-contributions}

\noindent {\bf Almost Singularly Optimal Asynchronous MST Algorithm.} Our main contribution is
 a randomized distributed MST  algorithm  that, with high probability, computes an MST in \emph{asynchronous}  $\mathcal{CONGEST}$ networks   and  takes $\tilde{O}(D^{1+\epsilon} + \sqrt{n})$ time and $\tilde{O}(m)$ messages,
 where  $n$ is the number of  nodes, $m$ the number of edges,  $D$ is the diameter of the network, and $\epsilon >0$ is an arbitrarily  small constant (both time and message bounds hold with high probability) (cf. Theorem \ref{the:mst-alg}). 
 Since $\tilde{\Omega}(D+\sqrt{n})$ and $\Omega(m)$ are respective time and message lower bounds for
 distributed MST construction in the $KT_0$ model, our algorithm  is message optimal (up to a $\polylog(n)$ factor) and almost time optimal (except for a $\tilde{O}(D^{\epsilon})$ factor).

\medskip

 \noindent {\bf Asynchronous MST in $KT_1$ in Sublinear Messages and Time.}
  Our result   answers an open problem raised in Mashregi and King~\cite{KMDISC19} (see also~\cite{MK21,KMDISC18}). They ask
   if there exists an asynchronous MST algorithm that takes sublinear
time if the diameter of the network is low, and has $\tilde{O}(m)$ message complexity. They remark that if such an algorithm exists, then it would improve their result giving better bounds for asynchronous MST in $KT_1$ $\mathcal{CONGEST}$.
 Our result answers their question in the affirmative 
  by giving the first known  asynchronous MST algorithm that has sublinear time (for all $D = O(n^{1-\delta})$, where $\delta > 0$ is an arbitrarily small constant) and uses $\tilde{O}(m)$ messages. Furthermore, as indicated in  Mashregi and King~\cite{KMDISC19}, this also yields the first asynchronous MST algorithm that is \emph{sublinear} in 
 \emph{both} time and messages in the
 $KT_1$ $\mathcal{CONGEST}$ model. More precisely, plugging our asynchronous MST algorithm in the result of
~\cite{KMDISC19}([Theorem 1.2]) gives an asynchronous MST algorithm that takes $\tilde{O}(D^{1+\epsilon}+ n^{1-2\delta})$ time 
 and $\tilde{O}(n^{3/2+\delta})$ messages  for any small constant $\epsilon > 0$ and for any $\delta \in [0,0.25]$ (cf. Theorem \ref{the:mst-alg-kt1}). 
 This  gives a tradeoff result between time and messages. In particular, setting $\delta = 0.25$ yields an asynchronous MST algorithm that has (almost optimal) time complexity
$\tilde{O}(D^{1+\epsilon}+ \sqrt{n})$ and message complexity $\tilde{O}(n^{7/4})$.

\medskip
 
 \noindent {\bf  Low Diameter Spanning Tree Construction.}
A key tool in our algorithm is the construction of a low diameter rooted spanning tree in asynchronous $\mathcal{CONGEST}$ 
 that has depth $\tilde{O}(D^{1+\epsilon})$ (for an arbitrarily small constant $\epsilon > 0$) in time $\tilde{O}(D^{1+\epsilon})$ time  and $\tilde{O}(m)$ messages.
 To the best of our knowledge, this is the first such construction that is almost singularly optimal in the asynchronous setting. This tree construction is of independent interest as it can also be used for {\em efficiently} (under both time and messages) performing tasks such as  upcast and downcast which are very common tools in distributed algorithms (these are described, for completeness, in Section \ref{sec:toolbox}). Informally, an upcast (using the tree) provides a feedback  (i.e., verification) to the broadcast (downcast) initiator such that  (1) the broadcast initiator knows when the broadcast terminates (based on
 acknowledgements from all nodes) and (2) the initiator can get compute a value based on the inputs of all the nodes (e.g., their sum). This verified broadcast is crucial in the asynchronous setting that allows the initiator to know
 when the broadcast has reached all nodes and thereafter proceed to the next step of the computation.
 
 We note that one could have used a BFS tree instead of a low-diameter tree. However,
 the best known BFS tree construction in the asynchronous setting is due to Awerbuch~\cite{A89} which takes $O(D^{1+\epsilon})$ time and $O(m^{1+\epsilon})$ messages (for arbitrarily small
 constant $\epsilon > 0$). This algorithm (which is deterministic) is not message optimal, unlike ours, and hence will only yield an MST algorithm with $O(m^{1+\epsilon})$ message complexity.
 Furthermore,  though our algorithm  does not compute a BFS  (but it is sufficient for MST purposes) and  is randomized, it is significantly simpler to understand and prove correctness for when compared to Awerbuch's algorithm. We also note that apart from the leader election and spanning tree primitives, the rest of the MST algorithm is deterministic.

\subsection{Additional Related Work}
\label{subsec:related-work}

The distributed MST problem has been studied intensively for the last four decades and there are several results known
in the literature, including several recent results, both for synchronous and asynchronous networks (including the ones mentioned in Section \ref{sec:intro}), see e.g., 
\cite{Elkin17,elkin-faster,PanduranganRS18,low-congestion-mst,haeupler2018round,universal-optimality-mst,MK21,MashreghiK17,
PanduranganRS17,pandurangan2019time} and the references therein.

We note that the results of this paper
and that of leader election of~\cite{KMPP21} (for asynchronous networks) as well as those of~\cite{PanduranganRS17,pandurangan2019time} and~\cite{Elkin17} (for synchronous networks)
assume the so-called \emph{clean network model}, a.k.a.\ $KT_0$~\cite{peleg-locality} (see Section \ref{subsec:model}), where nodes
do not have initial knowledge of the identity of their neighbors.
But the optimality of above results does not in general apply to the $KT_1$  model, where nodes have initial knowledge
of the identities of their neighbors.
It is clear that for time complexity by itself, the distinction between $KT_0$ and $KT_1$ does not matter (as one can simulate $KT_1$ in $KT_0$ in 
one round/time unit by each node
sending its ID to all its neighbors)
but it is significant when considering message complexity (as the just mentioned simulation costs $\Theta(m)$ messages).
Awerbuch et al.~\cite{AGVP90} show that $\Omega(m)$ is a message lower bound for broadcast (and hence for  construction of a spanning tree as well) in the $KT_1$ model,  
if one allows only (possibly randomized Monte Carlo) comparison-based algorithms, i.e., algorithms that can operate on IDs only by comparing them.
(We note that all  algorithms mentioned earlier in this subsection are comparison-based, including ours.)

On the other hand, for \emph{randomized non-comparison-based} algorithms, the message lower bound of $\Omega(m)$ does not apply in the $KT_1$ model. 
King et al.~\cite{KingKT15} presented a randomized, non-comparison-based Monte
Carlo algorithm in the $KT_1$ model for
MST construction in $\tilde{O}(n)$ messages ($\Omega(n)$ is a message lower bound) 
(see also~\cite{MashreghiK17}).
While this algorithm 
achieves $o(m)$ message complexity
(when $m = \omega(n \polylog n)$), it is \emph{not} time-optimal, as it takes time $\tilde{O}(n)$ rather than $\tilde{O}(D+\sqrt{n})$.
Algorithms with improved round complexity but worse message complexity, and more generally, trade-offs between time and messages, are shown in~\cite{GhaffariK18,gmyr}. We note that all these results are for \emph{synchronous} networks. 
As discussed  in Section \ref{sec:intro}, the works of~\cite{KMDISC19,KMDISC18,MK21} address asynchronous MST construction in $KT_1$ model and present algorithms that take $o(m)$ messages.

\section{Toolbox}
\label{sec:toolbox}
In this section, we present several procedures that are used as blackboxes in the current paper. As these procedures are either from other papers or minor variations of those in other papers, we merely mention what they do and their guarantees here. 
\paragraph*{Synchronization}
Synchronizers are mechanisms that allow nodes to run synchronous algorithms in an asynchronous network with some overhead, either in time or messages.

\textit{$\alpha$-synchronizer.}
An $alpha$-synchronizer, presented by Awerbuch~\cite{awerbuch1985complexity}, is a well known mechanism for nodes to run synchronous algorithms in an asynchronous network in the same running time (with a diameter overhead to time) while suffering a message overhead equivalent to the product of the run time of the synchronous algorithm and $O(m)$. 
Informally, when simulating some synchronous algorithm Alg, each node $v$ sends a ``pulse'' message to all its neighbors after all of $v$'s messages in the current round of Alg were acknowledged.  Thus, $v$'s neighbors can keep track of which pulse, or ``clock tick'', $v$ has simulated. Additionally, note that it takes $O(D)$ time to initialize the $\alpha$-synchronizer. 
A good description appears also in \cite{peleg-locality}. We know the following about an $\alpha$-synchronizer.

\begin{lemma}[Adapted from~\cite{peleg-locality}]
\label{lem:alpha-synchronizer}
Consider a graph $G$ with $n$ nodes, $m$ edges, and diameter $D$ in an asynchronous setting. The nodes of the graph may simulate a synchronous algorithm that takes $O(T)$ rounds and $O(M)$ messages in the synchronous setting by utilizing an $\alpha$-synchronizer. The resulting simulated algorithm takes $O(T+D)$ time and $O(M + Tm)$ messages and has termination detection.
\end{lemma}

\textit{$\beta$-synchronizer.}
A $\beta$-synchronizer is another type of synchronizer that reduces the message overhead at the expense of time. An assumption is made that there exists a spanning tree $\mathcal{T}$, rooted at some node $\mathcal{L}$, of depth $d$ overlaid on top of the original graph and that each node knows its parent and children in the tree, if any. Now, as with the $\alpha$-synchronizer, a synchronous algorithm that takes $O(T)$ rounds and $O(M)$ messages may be simulated in an asychronous network with the help of pulses. However, here each node sends a pulse to its parent once the current round is done and it has received pulses from each of its children in the tree. Once the root receives the pulse and finishes the current round, it broadcasts a message to move to the next round along the tree. The resulting simulated algorithm takes $O(T\cdot d)$ time and $O(M + Tn)$ messages.

\begin{lemma}[Adapted from~\cite{peleg-locality}]
\label{lem:beta-synchronizer}
Consider a graph $G$ with $n$ nodes in an asynchronous setting. Assume that there exists a rooted spanning tree $\mathcal{T}$ of depth $d$ overlaid on $G$ such that each node knows its parent and children, if any, in the tree. The nodes of the graph may simulate a synchronous algorithm that takes $O(T)$ rounds and $O(M)$ messages in the synchronous setting by utilizing a $\beta$-synchronizer over $\mathcal{T}$. The resulting simulated algorithm takes $O(T \cdot d)$ time and $O(M + Tn)$ messages and has termination detection.
\end{lemma}

Notice that both $\alpha$- and $\beta$-synchronizers can be used by nodes to enact a type of global round counter up to any number that can be encoded using $O(\log n)$ bits.


\paragraph*{Leader Election}
We make use of the leader election procedure, call it Procedure~$\LeaderElection$, of Kutten et al.~\cite{KMPP21} to elect a leader with high probability. Adapting Theorem~11 to this setting, we have the following lemma.
Note that in the course of the procedure, all nodes are woken up but such information was not mentioned in the theorem statement in~\cite{KMPP21}, so we add it here.

\begin{lemma}[Theorem~11 in~\cite{KMPP21}]
\label{lem:leader-election} 
Procedure~$\LeaderElection$ solves leader election with termination detection with high probability in any arbitrary graph with $n$ nodes, $m$ edges, and diameter $D$ in $O(D + \log^2 n)$ time with high probability using $O(m \log^2 n)$ messages with high probability in an asynchronous system with adversarial node wake-up. At the end of the procedure, all nodes are awake.
\end{lemma}

\paragraph*{Operations on a Fragment}
In the course of our algorithm, we reach a situation where the graph $G$ is partitioned into a set of disjoint trees 
(called fragments), each with a distinct root, an associated fragment ID, and an associated cluster ID (which may be different from its fragment ID). Each node knows its parent and children in the fragment, if any. We now describe some common operations that are to be performed on such trees.

Consider a tree $T$ spanning a subset of the nodes of $G$, oriented towards a distinct root $R$. Let the tree have fragment ID $F$, known to all nodes in $T$. Furthermore, all nodes of $T$ have the same cluster ID, say $C$, which may or may not be equal to $F$. Let $size(T)$ and
$depth(T)$ denote the number of vertices and the depth of $T$, respectively.

\textit{Broadcast on a Fragment.}
Suppose a message $M$, originating at the root $R$, must be distributed to all nodes of the tree.
Procedure~$\FragmentBroadcast$ performs this operation in a straightforward manner. 
The root $R$ sends $M$ to all its neighbors. Intermediate nodes receiving $M$ on some round forward it to all their children in $T$ in the next round. 

To ensure termination detection, the procedure then performs a {\em convergecast} of acknowledgements on $T$ as follows. Each leaf, upon receiving $M$, sends back an ``ack'' message.
Each intermediate node waits until it receives an ``ack'' from all its children, and then sends an ``ack'' to its parent. The operation terminates once the root receives an ``ack'' from all its children. 

\begin{lemma}
\label{lem:frag-broadcast}
Procedure~$\FragmentBroadcast$, run by nodes in the tree $T$, performs broadcast of a message originating at the root of~ $T$ with termination detection in $O(depth(T))$ time and $O(size(T))$ messages.
\end{lemma}

\textit{Upcast on a Fragment.}
Suppose $k$ distinct and uncombinable messages, originating at arbitrary locations in the tree, must be gathered to the root $R$. 
Procedure~$\Upcast$ performs this operation in a straightforward manner. Each node in the tree pipelines the messages it has seen upwards in the tree (towards $R$), in some arbitrary order. 

 We assume that $R$ knows the number $k$ of such messages it expects to receive and ensure this is true everywhere the procedure is called. Thus, $R$ knows when it has received all $m$ messages. 
 To ensure termination detection, the procedure then performs $\FragmentBroadcast$.

\begin{lemma}
\label{lem:upcast}
Procedure~$\Upcast$, run by nodes in the tree $T$, performs upcasting of $k$ distinct messages with termination detection in $O(k + depth(T))$ time and $O(k \cdot depth(T))$ messages.
\end{lemma}

\textit{Downcast on a Fragment.}
Suppose $k$ distinct and uncombinable messages $M_1,\ldots, M_k$, originating at the root $R$, must be distributed to arbitrary destinations $w_1,\ldots,w_k$ in the tree, respectively.
Procedure~$\Downcast$ performs this operation in a straightforward manner. In each round $i$, $R$ sends the pair $(M_i,w_i)$ to its neighbor on the unique $R$-$w_i$ path in $T$. Intermediate nodes receiving a pair $(M_i,w_i)$ on some round forward it towards $w_i$ in the next round. (Note that tie-breaking is not required.)

To ensure termination detection, the procedure then performs a convergecast of acknowledgements, backtracking on the subtree $T'$ marked by the downcast messages; namely, each intermediate node that received $\ell$ messages from its parent and forwarded $\ell_j$ messages to its child $x_j$ expects ``ack - $\ell_j$'' from $x_j$. After receiving all such ``ack'' messages from its children, it sends ``ack - $\ell$'' to its parent. The root detects termination upon receiving ``ack'' messages from all relevant children.

\begin{lemma}
\label{lem:downcast}
Procedure~$\Downcast$, run by nodes in the tree $T$, performs downcasting of $k$ distinct messages with termination detection in $O(k + depth(T))$ time and $O(k \cdot depth(T) + size(T))$ messages.
\end{lemma}

\textit{Finding MOE of a Fragment.}
 Informally, minimum outgoing edge (MOE) out of $T$ is the least weight edge out of $T$ to a node with a different cluster ID (i.e, $\neq C$). Formally, it is a tuple $\langle u, v, C, C' \rangle$ such that edge $(u,v)$ is the MOE from $T$ where $u \in T$ with cluster ID $C$ and $v \notin T$ with cluster ID $C' (\neq C)$. Note that nodes not belonging to $T$ but adjacent to $T$ may have the same cluster ID $C$ as the nodes of $T$, and as such it is possible for $T$ to not have any MOE. Yet another application of Wave\&Echo, taken from the algorithm of \cite{GallagerHS83},
results in $R$ being made aware of the MOE of $T$ if such exists. Let us call this module procedure~$\FindMOE$.

\begin{lemma}
\label{lem:finding-moe}
Procedure~$\FindMOE$, when run by the nodes of a tree $T$ with distinct root $R$, and cluster ID $C$, results in $R$ knowing the minimum outgoing edge from $T$, if one exists, where only edges to nodes with a cluster ID $\neq C$ are considered outgoing edges, in $O(depth(T))$ time and $O(\sum_{u \in T} deg(u))$ messages, where $depth(T)$ is the depth of $T$ and $deg(u)$ is the degree of node $u$. Furthermore, every node participating in procedure~$\FindMOE$ can detect termination.
\end{lemma}

\textit{Size Calculation of a Fragment.}
We make use of a known tool (essentially a known application of Wave\&Echo, see PIF in \cite{segall1983distributed}), to be run by the nodes of the tree and result in $R$ being made aware of how many nodes (including itself) belong to $T$.
Let us call this Procedure~$\TreeCounting$.

\begin{observation}
\label{lem:tree-counting}
Procedure~$\TreeCounting$, when run by the nodes of a tree $T$ with distinct root $R$, results in $R$ knowing the total number of nodes in $T$ in $O(depth(T))$ time and $O(size(T))$ messages, where $depth(T)$ is the depth of $T$ and $size(T)$ is the number of nodes in $T$. Furthermore, nodes participating in procedure~$\TreeCounting$ can detect termination.
\end{observation}

\textit{Diameter Calculation of a Fragment.}
Another known application of Wave\&Echo allows
$R$ to calculate the diameter of the tree $T$, let us call that 
Procedure~$\DiameterCalculation$.

\begin{observation}
\label{lem:diameter-calculation}
Procedure~$\DiameterCalculation$, when run by the nodes of a tree $T$ with distinct root $R$, results in $R$ knowing the diameter of $T$ in $O(depth(T))$ time and $O(size(T))$ messages, where $depth(T)$ is the depth of $T$ and $size(T)$ is the number of nodes in $T$. Furthermore, nodes participating in procedure~$\DiameterCalculation$ can detect termination.
\end{observation}

\section{Low Diameter Spanning Tree Algorithm}
\label{sec:lowDiamSpanningTree}
Let us now describe a novel algorithm for constructing a low diameter spanning tree in a time-efficient and (near) message-optimal manner in an \emph{asynchronous} network. This serves as a crucial ingredient for our MST algorithm
of Section \ref{sec:mst}.

\subsection{Randomized Low Diameter Decomposition (MPX)}
\label{ss:MPX}

Let $\overline{G}=(\overline{V},\overline{E})$ be any (undirected, unweighted) graph with $\overline{n} \leq n$ nodes and $\overline{m} \leq m$ edges; in particular, $\overline{G}$ can be different from the communication graph.
A probabilistic $(\beta, r$) \emph{low diameter decomposition} of $\overline{G}$ is a partition of $\overline{V}$ into disjoint node sets $\overline{V}_1, \ldots, \overline{V}_t$ called \emph{clusters}. The partition satisfies (1) each cluster $\overline{V}_i$ has strong diameter $r$, i.e., $dist_{\overline{G}[V_i]}(u,v) \leq r$ for any two nodes $u,v \in \overline{V}_i$, and (2) the probability that an edge $e \in \overline{E}$ is an \intercluster\ edge  (that is, the endpoints of $e$ are in different clusters) is at most $\beta$.

\subparagraph*{MPX Decomposition in Synchronous $\mathcal{CONGEST}$} Let us describe a simple distributed variant of the MPX decomposition algorithm of Miller et al. \cite{MPX13} --- Procedure~$\MPX$ --- executed in a synchronous setting with simultaneous wakeup on graph $\overline{G}$. 
In Subsect. \ref{sec:rootedSpanningTree}, we execute the algorithm on virtual cluster graphs (where each node is in fact a set of nodes in the communication graph $G$) and also describe the distributed simulation required to do so.

Let $\delta_{max} = \lfloor 2 \cdot \frac{\ln n}{\beta} \rfloor$. Initially, each node $v \in \overline{V}$ draws a random variable $\delta_v$ from the exponential random distribution with parameter $\beta$ and sets its \emph{start-time} variable $S_v$ to $\max\{1,\delta_{max} - \lfloor \delta_v \rfloor\}$.
Procedure~$\MPX$ guarantees the following through simple flooding: (1) each node $v \in \overline{V}$ is assigned to the cluster of the node $u = argmin_{w \in \overline{V}} \{(dist_{\overline{G}}(v,w) + S_w, id_w)\}$
and (2) each cluster has a spanning tree of depth at most $\delta_{max}$. (Each node locally keeps information about the edge to its parent in the spanning tree. In other words, the spanning tree is oriented towards the root.)

More precisely, the ``simple flooding'' is done in $\delta_{max}+1$ rounds. Initially, all nodes are \emph{unassigned}. In round $i$, each newly-assigned node $v$ (i.e., assigned in round $i-1$) sends to its neighbors a message containing the ID of the cluster leader. 
Other assigned nodes do nothing. 
Finally, for each unassigned node $v$, let $M_{id}$ be the set containing all received IDs, as well as $id_v$ if $S_v=i$. If $M_{id}$ is the empty set, $v$ does nothing. Otherwise, $v$ assigns itself to the cluster of the node $u$ with the lexicographically smallest ID in $M_{id}$. If $u \neq v$, $v$ keeps the edge (an arbitrary one if there are multiple such edges) along which it receives $id_u$ as the edge to its parent. (Note that this spanning tree guarantees that the cluster is connected and has strong diameter at most $\frac{4 \ln n}{\beta}$.)

\subparagraph*{Analysis} 
The following lemmas are known results from \cite{MPX13,HW16,CDHHLP18,CDHP20}. We first provide definitions and an auxiliary lemma (see Lemma \ref{lem:MPXArrivals}) followed by proofs of Lemmas \ref{lem:mpxDecomp} and \ref{lem:MPXDiameter}.

Consider some fixed execution of the algorithm and node $v \in \overline{V}$. Then $D_u = S_u + dist(u,v) - 1 = \delta_{max} - \lfloor \delta_u \rfloor + dist(u,v) - 1$ denotes the \emph{(arrival) round} of $u$, that is, the first round in which $v$ can receive a message from $u$'s cluster. 
For every integer $1 \leq j \leq n$, let $z_j$ be the node with the $j$th smallest arrival round in the execution. For every integer $1 \leq k \leq n$, let  $S_k=\{z_1,\ldots,z_k\}$. Building upon these definitions, for a node $v \in \overline{V}$, positive integers $1 \leq k,r \leq n$, let $\EV_{v,k,r}$ denote the event that after the execution of the algorithm, $D_{z_{k+1}} - D_{z_1} \leq r$.

\begin{lemma}
\label{lem:MPXArrivals}
For any node $v \in \overline{V}$ and positive integers $1 \leq k,r \leq n$, 
$$\Pr(\EV_{v,k,r}) ~\le~ (1 - \exp(- (r+1) \beta))^k$$
\end{lemma}

\begin{proof}
We condition on $S_k$ and $D^*=D_{z_{k+1}}$. The proof is based on first showing the stated upper bound on the probability of $\EV_{v,k,r}$ conditioned on $S_k$ and $D^*$, and then applying the law of total probability to derive the lemma statement. We next describe the first half of the proof in more detail.

For any integer $i \geq 1$, let $c_{z_i} = \delta_{max} + dist(z_i,v) - 1$. We have
$\Pr(\EV_{v,k,r} \; | \; S_k, D^*)\le p$ for
$$p ~=~ \Pr\left(\bigwedge_{i=1}^k [D^* - D_{z_i} \leq r]\right) ~=~ \Pr\left(\bigwedge_{i=1}^k [\delta_{z_i} \leq r + 1 +c_{z_i} - D^*]\right) ~=~ \prod_{i=1}^k \Pr(\delta_{z_i} \leq r + 1 +c_{z_i} - D^*),$$ 
where the last equality holds since the random variables $\delta_{z_i}$ are independent. 
Next, note that $D^* \geq D_{z_i}$ for any integer  $1 \leq i \leq k$, and thus $\Pr(\lfloor \delta_{z_i} \rfloor \geq c_{z_i} - D^*) = 1$. Hence, $\Pr(\delta_{z_i} \geq c_{z_i} - D^*) = 1$ and 
$$p ~=~ \prod_{i=1}^k \Pr(\delta_{z_i} \leq r +1 + c_{z_i} - D^* \; | \; \delta_{z_i} \geq c_{z_i} - D^*).$$ 
Finally,  
$$p ~\leq~ \prod_{i=1}^k \Pr(\delta_{z_i} \leq r+1) ~=~ \prod_{i=1}^k (1 - \exp(- (r+1) \beta)) ~=~ (1 - \exp(- (r+1) \beta))^k$$
where the inequality holds by the memorylessness of the exponential distribution.
\end{proof}

\begin{lemma}
\label{lem:mpxDecomp}
Procedure~$\MPX$ computes a $(2 \beta$, $\frac{4 \ln n}{\beta})$ low-diameter decomposition of $\overline{G}$ w.h.p. in $O(\frac{\ln n}{\beta})$ time and $O(m \frac{\ln n}{\beta})$ messages in the synchronous setting.
\end{lemma}

\begin{proof}
We first note for any node $v \in \overline{V}$, $\Pr[\lfloor \delta_v \rfloor > \delta_{max}] = \Pr[\delta_v > \frac{2 \ln n}{\beta}] = \exp(- 2 \ln n) = \frac{1}{n^2}$. Hence, by union bound, $\lfloor \delta_v \rfloor \leq \delta_{max}$ for every node $v \in \overline{V}$ with high probability. We hereafter exclude this unlikely event and assume $\delta_{max} \geq \max_{v \in \overline{V}} \{\lfloor \delta_v \rfloor\}$. This implies that all nodes belong to a cluster.

Next, note that by the algorithm description, each cluster is spanned by a tree of depth at most $\frac{2 \ln n}{\beta}$. Hence, all clusters have strong diameter at most $\frac{4 \ln n}{\beta}$.
Finally, an edge is cut if its two endpoints $u$ and $v$ are in different clusters. This implies that for node $v$ (without loss of generality), the two smallest arrival rounds differ by at most 1, which corresponds to event $\EV_{v,1,1}$. By Lemma \ref{lem:MPXArrivals}, $\Pr(\EV_{v,1,1}) \leq (1 - \exp(- 2\beta)) \leq 2 \beta$. The lemma follows.
\end{proof}

From the low diameter decomposition computed by Procedure~$\MPX$ (or in fact, from any partition $\mathcal{P}$ of $\overline{V}$ into disjoint node sets $\overline{V}_1, \ldots, \overline{V}_t$), one can define a cluster graph $\overline{G}^*= (\overline{V}^*, \overline{E}^*)$, as follows. Its node set $\overline{V}^*$ $ = \{\overline{V}_1, \ldots, \overline{V}_t\}$ consists of cluster nodes, one for each cluster $\overline{V}_i$ of the decomposition, and two cluster nodes $\overline{V}_i$ and $\overline{V}_j$ are adjacent in $\overline{G}^*$ if there exist two nodes $w,w'$ in $\overline{V}$ such that $w\in\overline{V}_i$, $w'\in\overline{V}_j$ and $(w,w')\in\overline{E}$. We call $\overline{G}^*$ the \emph{cluster graph induced by $\mathcal{P}$.} 

\begin{lemma}
\label{lem:MPXDiameter}
For any positive integer $k \geq 1$, if the diameter of $\overline{G}$ satisfies $\overline{D} \geq k \frac{\ln^2 n}{\beta^4}$, then the diameter of the cluster graph $\overline{G}^*$ is at most $2 \beta \overline{D}$, with probability at least $1 - \frac{1}{n^{k-2}}$.
\end{lemma}

\begin{proof}
Again, we assume $\delta_{max} \geq \max_{u \in \overline{V}} \{\delta_u\}$, which holds with high probability. For any node $v \in \overline{V}$, let $C_v$ denote the cluster containing $v$ after the execution of the algorithm.

Consider any two nodes $u,v \in \overline{V}$ such that $l = dist_{\overline{G}}(u,v) > 3 \beta \overline{D}$. (Note that if $l \leq 3 \beta \overline{D}$, then $dist_{\overline{G}^*}(C_u,C_v) \leq 3 \beta \overline{D}$.)
Let $(w_1,\ldots,w_{l+1})$ be the shortest path  between $u$ and $v$ in $\overline{G}$ (where $w_1 = u$ and $w_{l+1} = v$). Moreover, for any integer $i \in [1,l]$, let $X_i$ be the indicator random variable of $w_i$ and $w_{i+1}$ being in the same cluster. Then, the random variable $X = \sum_{i=1}^{l} X_i$ is an upper bound on $dist_{\overline{G}^*}(C_u,C_v)$. By Lemma \ref{lem:mpxDecomp}, each edge is an inter-cluster edge with probability at most $2\beta$. Hence, by the linearity of expectation, $E[X] \leq 2\beta l$.

Next, let us provide a concentration bound for $X$ by showing that the random variables $X_i$ are only locally dependent. First, for any two integers $i, j \in [1,l]$ such that $|i - j| > \lfloor 4 \frac{\ln n}{\beta} \rfloor$, $X_i$ and $X_j$ are independent (since the same node cannot affect $w_i$ and $w_j$ with our choice of $\delta_{max}$). Then, we can color the random variables $\{X_i\}_{i=1,\ldots,l}$ using $\chi = \lfloor 4 \frac{\ln n}{\beta} \rfloor$ --- by coloring $X_i$ with $i \mod (\chi + 1)$ --- such that variables with the same color are independent. In other words, the random variables $X_i$ are only locally dependent and thus we can apply a specific Chernoff-Hoeffding bound (Theorem 3.2 from \cite{DP09}): $\Pr(X \geq E[X] + t)  \leq \exp(-2t^2 / (\chi \cdot l))$. Hence, $\Pr(X \geq 3 \beta l)  \leq \exp(-2 (\beta l)^2 / (\chi \cdot l)) \leq \exp(-2 \beta^2 l /\chi)$.
Since $l > 3 \beta \overline{D} > 3 k \frac{\ln^2 n}{\beta^3}$, $\Pr(X \geq 3 \beta l) \leq \exp(- \frac{3}{2} k \ln n) \leq \frac{1}{n^k}$. By taking a union bound over all $n^2$ possible pairs of nodes $u,v \in \overline{V}$, the lemma statement follows.
\end{proof}

\subsection{Rooted Spanning Tree}
\label{sec:rootedSpanningTree}
Let us now describe an asynchronous distributed algorithm to construct a low diameter rooted spanning tree, given a pre-specified root, in a time-efficient and (near) message-optimal manner --- see Theorem \ref{thm:spanningTreeConstruction}. We assume that each node knows whether it is the pre-specified root prior to the start of the algorithm. We also assume initially that the diameter of the original graph, $D$, is known to the nodes. We explain how to remove this assumption at the end of the section.

\begin{theorem}
\label{thm:spanningTreeConstruction}
Given a graph $G$ with $n$ nodes, $m$ edges and diameter $D$, as well as a distinguished node $R$, and a constant parameter $1 \geq \epsilon>0$, the asynchronous distributed Procedure~$\SpanningTreeConstruction(\epsilon)$ computes an $\Tilde{O}(D^{1+\epsilon})$-diameter spanning tree rooted in $R$ with termination detection, using $\Tilde{O}(D^{1+\epsilon})$ time with high probability and $\Tilde{O}(m)$ messages with high probability.
\end{theorem} 

\subparagraph*{Brief Description.} We construct the low diameter spanning tree in a two stage process. The first stage consists of building a sequence of increasingly coarser partitions of $G=(V,E)$. Each partition decomposes $V$ into disjoint node sets, called clusters, with strong diameter $\Tilde{O}(D^{1 + \epsilon})$; in fact, each cluster $C$ is spanned by a tree $\CLT(C)$ of depth $\Tilde{O}(D^{1 + \epsilon})$. (Unlike in Subsect.\ \ref{ss:MPX}, this spanning tree is oriented away from the root.) The unique cluster containing the root node $R$ will be denoted $C_R$. 
The cluster graph induced by the final partition (defined in Subsect.\ \ref{ss:MPX}) has diameter $\Tilde{O}(1)$. These partitions are obtained by simulating the synchronous MPX decomposition algorithm (see Subsect.\ \ref{ss:MPX}) on $G$, then on the obtained cluster graph, and so on, for $\imax = \lceil \log_{1/(3\beta)} D \rceil$ times (where $\beta = \ln^{-\frac{1}{\epsilon'}} n$ and $\epsilon' \leq 1$ is to be derived in the analysis). 
In the second stage, we construct a breadth first search (BFS) tree $\BFST$ over the final cluster graph of phase 1, where the cluster $C_R$ containing the pre-specified root $R$ serves as the root of the BFS tree. We then use $\BFST$ to decide which edges of the original graph should be kept to obtain the desired rooted spanning tree $\SPT$ of $G$ with depth $\Tilde{O}(D^{1 + \epsilon})$.

\subparagraph*{Detailed Description.}
Consider the initial graph $G(V,E) = G_0(V_0,E_0)$ and the initial trivial partition $\mathcal{P}_0$ in which each node $v \in V$ is its own cluster.
\begin{itemize}
    \item \textbf{Stage 1:} The first stage consists of $\imax = \lceil \log_{1/(3\beta)}D \rceil$ phases, where $\beta = \ln^{-1/\epsilon'} n$ and we assume $\ln \ln n \geq 2 \epsilon' \ln 3$. (If $\ln \ln n \leq 2 \epsilon' \ln 3 \leq 2 \ln 3$, then constructing a low diameter spanning tree efficiently is trivial.) Phase $i$ starts with a partition $\mathcal{P}_{i-1}$ of $V$ and the cluster graph induced by $\mathcal{P}_{i-1}$ is denoted by $G_{i-1}(V_{i-1},E_{i-1})$. We simulate one instance of Procedure~$\MPX$ (with parameter $\beta$) on $G_{i-1}$ in an asynchronous setting by running an $\alpha$-synchronizer between clusters, and within each cluster $C$, using the spanning tree $\CLT(C)$ to simulate the behavior of each cluster node of $V_{i-1}$. (Note that this well-known synchronizer is described in more detail in Section \ref{sec:toolbox}.) 
    More precisely, the root of the spanning tree $\CLT(C)$ simulate the behavior of cluster $C$ (in the simulated Procedure~$\MPX$). To send a (same) message to its adjacent clusters, $C$ broadcasts along $\CLT(C)$. To receive the message with the minimum ID (which is sufficient information for Procedure~$\MPX$), $C$ convergecasts along $\CLT(C)$.

    The output is a partition $\mathcal{P}_i^*$ of $V_{i-1}$ into disjoint (cluster node) sets $U_1,\ldots,U_t$ 
    such that each $U_j$ has a spanning tree $\SUPT_j$ of depth $O(\frac{\ln n}{\beta})$. We transform $\mathcal{P}_i^*$ into a partition $\mathcal{P}_i$ of $V$, the node set of the \emph{original} graph, into disjoint node sets $W_1,\ldots,W_t$, such that each $W_j$ has a spanning tree $\CLT(W_j)$ of depth $O((\frac{\ln n}{\beta})^i)$. (In fact, we only show how to compute the spanning trees $\CLT(W_j)$, which induces the node sets $W_j$.) 
    
    To transform $\mathcal{P}_i^*$ to $\mathcal{P}_i$, we use a simple Procedure~$\Transform$, sketched next. Recall that each cluster node in $U_j$ keeps information about its parent in the spanning tree $\SUPT_j$. Procedure~$\Transform$ consists of $2\frac{\ln n}{\beta}$ iterations. Each cluster node keeps an iteration counter and these counters are kept locally synchronized by running an $\alpha$-synchronizer between cluster nodes. In the first iteration, the root cluster node $C_R$ sends its ID to each adjacent cluster node $C$ (which is its child in $\SUPT_j$) over the edges of the set $E_{inter}=\{(u,w)\in E(G) \mid u\in C_R, w\in C\}$, namely, all (original) inter-cluster edges between $C_R$ and $C$. (Note that in fact, $C_R$ sends its ID to all adjacent cluster nodes, but cluster nodes which are not children of $C_R$ simply ignore that message.) Among these inter-cluster edges, every child cluster node $C$ keeps $(u^*,w^*) = argmin_{(u,w) \in E_{inter}} \{id_w\}$, i.e., the edge whose endpoint $w$ in $C$ has the minimum ID. 
    
    Cluster node $C$ then reorients its tree $\CLT(C)$ to be rooted in $w$ (and the \intercluster\ edge is oriented towards $w$, i.e., from parent to child). 
    In the next iteration, each $C$ sends the ID of $R$ to its children cluster nodes, if they exist, which in turn reorient their tree in the same fashion. After all iterations are done, the ``combined'' spanning tree $\CLT(W_j)$ is completed, and a simple broadcast allows all nodes in the newly computed cluster $W_j$ to move on to the next phase. (Note that $\CLT(W_j)$ is oriented from the root outwards.)

    \item \textbf{Stage 2:} At the end of stage 1, the final partition decomposes $V$ into clusters with strong diameter $\tilde{O}(D^{1+\epsilon})$ and induces a cluster graph $G_f(V_f,E_f)$ of diameter $O(\log^{2+4/\epsilon'} n)$; in fact, each cluster $C$ is spanned by a tree $\CLT(C)$ of depth $\tilde{O}(D^{1+\epsilon})$. During stage 2, the naive synchronous BFS tree construction algorithm (based on flooding, see \cite{peleg-locality}) is simulated on $G_f$ for $O(\log^{2+4/\epsilon'} n)$ rounds, where the designated root in $V_f$ is the cluster $C_R$ that contains the pre-specified root in $V$. 
    Once again, this is done by running an $\alpha$-synchronizer between clusters, and within each cluster, using the spanning tree $\CLT(C)$ to simulate the behavior of each cluster node $C$. After computing the BFS tree $\BFST$ on $G_f$, we use Procedure~$\Transform$---but this time for $O(\log^{2+4/\epsilon'} n)$ rounds---to compute a spanning tree $\SPT$ of $G$, similarly to stage 1.
    This final output $\SPT$ is a $\tilde{O}(D^{1+\epsilon})$ diameter spanning tree of $G$.
\end{itemize}

\subparagraph*{Analysis.} Lemma \ref{lem:diameterOfClustersAndClusterGraph} upper bounds, for each phase, the diameter of the cluster graph as well as that of the partition's clusters. Corollary \ref{cor:finalClusterGraph} is obtained from Lemma \ref{lem:diameterOfClustersAndClusterGraph} by considering the last phase. After which, we prove Theorem \ref{thm:spanningTreeConstruction} using Lemma \ref{lem:diameterOfClustersAndClusterGraph} and Corollary \ref{cor:finalClusterGraph}. 

\begin{lemma}
\label{lem:diameterOfClustersAndClusterGraph}
For each phase $1 \leq i \leq \imax$, (1) $\diam(G_{i-1})=\max\{ (3 \beta)^{i-1} D, O(\log^{2+4/\epsilon'} n)\}$ w.h.p., and (2) each cluster $C$ of the partition $\mathcal{P}_{i-1}$ is spanned (in the original graph $G$) by a tree $\CLT(C)$
with $\diam(\CLT(C)) = (\frac{5 \ln n}{\beta})^{i-1}$. 
\end{lemma}

\begin{proof}
By induction on $i$. The base case, $i = 1$, holds trivially. 

Next, consider some $i \geq 1$ for which the inductive hypothesis holds, i.e., $\diam(G_{i-1})=\max\{ (3 \beta)^{i-1} D, O(\log^{2+4/\epsilon'} n)\}$ w.h.p. and each cluster node $C$ of the partition $\mathcal{P}_{i-1}$ is spanned (in the original graph $G$) by a tree $\CLT(C)$ with $\diam(\CLT(C))= (\frac{5 \ln n}{\beta})^{i-1}$. Running Procedure~$\MPX$ on $G_{i-1}$ yields a $(2\beta, \frac{4 \ln n}{\beta})$ low-diameter decomposition of $G_{i-1}$. In fact, each super cluster $C'$ of this decomposition on $G_{i-1}$ is spanned (in the cluster graph $G_{i-1}$) by a tree $\CLT(C')$ of diameter $\frac{4 \ln n}{\beta}$. 
Hence, the ``combined'' spanning tree computed by Procedure~$\Transform$ for the ``analog'' $C''$ of cluster $C'$ on $G$, which is a cluster of the newly constructed $G_i$, has diameter $\diam(\CLT(C''))= (\frac{4 \ln n}{\beta} + 1) \cdot (\frac{5 \ln n}{\beta})^{i-1} \leq  (\frac{5 \ln n}{\beta})^{i}$. Next, the diameter of $G_i$ is the same as that of the cluster graph $H$ induced by partition $\mathcal{P}_i^*$. By Lemma \ref{lem:MPXDiameter}, the diameter of $H$ is $\max\{ (3 \beta)^{i} D, O(\log^{2+4/\epsilon'} n)\}$ w.h.p., and thus the lemma statement holds.
\end{proof}

\begin{corollary}
\label{cor:finalClusterGraph}
At the end of phase $\imax$, (1) $\diam(G_{\imax})= O(\log^{2+4/\epsilon'} n)$ w.h.p., and (2) each cluster $C$ of the partition $\mathcal{P}_{\imax}$ is spanned (in the original graph $G$) by a tree $\CLT(C)$ with $\diam(\CLT(C))= \tilde{O}(D^{1+\epsilon})$. 
\end{corollary}

\begin{proof}
By Lemma \ref{lem:diameterOfClustersAndClusterGraph} (and applying one extra induction step), the diameter of $G_{\imax}$ is $D_f = \max\{ (3 \beta)^{\imax} D, O(\log^{2+4/\epsilon'} n)\}$ and each cluster $C$ of the partition $\mathcal{P}_{\imax}$ is spanned in $G$ by a tree $\CLT(C)$ of depth $d_f = (\frac{5 \ln n}{\beta})^{\imax}$. Since $\imax = \lceil \log_{1/(3\beta)} D \rceil$, we have that $(3 \beta)^{\imax} \leq 1/D$, so $D_f = O(\log^{2+4/\epsilon'} n)$. Moreover, by going through the computations, we get: 
\begin{align*}
d_f &=~  \exp(\imax \ln(5 \ln^{1+1/\epsilon'} n)) 
~\leq~  (5 \ln^{1+1/\epsilon'} n)  \exp\left(\frac{\ln D \ln(5 \ln^{1+1/\epsilon'} n)}{\ln(\frac{1}{3} \ln^{1/\epsilon'} n) }\right)  
&\\
&=~ (5 \ln^{1+1/\epsilon'} n)   \exp\left(\ln D \cdot \frac{\ln 5 + (1+1/\epsilon') \ln\ln n}{\frac{1}{\epsilon'} \ln\ln n - \ln 3}\right) &\\
&=~  (5 \ln^{1+1/\epsilon'} n)  \exp\left(\ln D \cdot \left(1 + \frac{\ln 5 + \ln 3 + \ln\ln n}{ \frac{1}{\epsilon'} \ln\ln n - \ln 3}\right)\right)
&\\
&\leq~  (5 \ln^{1+1/\epsilon'} n) \exp(\ln D \cdot (1 + 2\epsilon' \ln 15 )) 
~\leq~  (5 \ln^{1+1/\epsilon'} n) ~ D^{1 + \epsilon}~, 
\end{align*}
where, in order to make the last inequality hold, Procedure~$\SpanningTreeConstruction(\epsilon)$ selects $\epsilon' \leq \epsilon / (2\ln 15)$.
\end{proof}

\begin{proof}[Proof of Theorem \ref{thm:spanningTreeConstruction}]
The correctness of the first stage follows from that of the simulation (using an $\alpha$-synchronizer between clusters), Procedure~$\MPX$ and Procedure~$\Transform$. Next, let us show the time and message complexity of the first stage.
During each phase $1 \leq i \leq \imax$, Procedure~$\MPX$ is simulated on $G_{i-1}$ for $O(\frac{\log n}{\beta}) = \tilde{O}(1)$ rounds. Hence, each cluster $C$ simulates $\tilde{O}(1)$ rounds. In each round, the cluster broadcasts once over the cluster's spanning tree $\CLT(C)$, sends one message per inter-cluster edge over to adjacent clusters, and convergecasts once over $\CLT(C)$. By Lemma \ref{lem:diameterOfClustersAndClusterGraph}, $\CLT(C)$ has depth $\tilde{O}(D^{1+\epsilon})$. Hence, each round of Procedure~$\MPX$ is simulated in at most $\tilde{O}(D^{1+\epsilon})$ time and using $O(m)$. Adding up over all phases results in $\tilde{O}(D^{1+\epsilon})$ time and $\tilde{O}(m)$ messages.
Note that running an $\alpha$-synchronizer (between the clusters) induces only an $\tilde{O}(1)$ message overhead per (\intercluster) edge over all rounds, but no time overhead. Thus Procedure~$\MPX$ is simulated in $\tilde{O}(D^{1+\epsilon})$ time and using $\tilde{O}(m)$ messages. 
Similarly, in Procedure~$\Transform$, each cluster $C$ simulates $\tilde{O}(1)$ rounds. In each round, the cluster broadcasts twice over the cluster's spanning tree $\CLT(C)$, sends one message per inter-cluster edge over to adjacent clusters, and convergecasts twice over $\CLT(C)$ (where the additional broadcast and convergecast allows to reorient $\CLT(C)$). Therefore, it can be seen that Procedure~$\Transform$ also takes $\tilde{O}(D^{1+\epsilon})$ time and uses $\tilde{O}(m)$ messages.
Finally, the first stage has at most $\imax = \tilde{O}(1)$ phases, and thus takes $\tilde{O}(D^{1+\epsilon})$ time and uses $\tilde{O}(m)$ messages.

By Corollary \ref{cor:finalClusterGraph}, the final cluster graph has a diameter of $O(\log^{2+4/\epsilon'} n)$. Given that, the correctness of the second stage follows from that of the simulation (using an $\alpha$-synchronizer between clusters), the naive synchronous BFS tree construction algorithm and Procedure~$\Transform$. As for the time and message complexity, the same approach (used for stage 1 above) shows that the second stage takes $\tilde{O}(D^{1+\epsilon})$ time and uses $\tilde{O}(m)$ messages. 
\end{proof}

\textbf{Removing the Requirement of the Knowledge of $D$.}
In the previously described algorithm, we assumed that each node knew the value of $D$, the diameter of the original graph. This assumption can be removed by having each node guess the value of $D = 2^1, 2^2, \ldots$ until we arrive at the correct guess (an at most $2$-approximation of $D$). 

An issue that must be addressed, however, is that nodes need some way to determine whether they have correctly guessed the value of $D$ or not. This can be done at the end of the second stage. Recall that the naive synchronous BFS tree construction is simulated for $O(\log^{2+4/\epsilon'} n) = \Tilde{O}(1)$ rounds. If the estimate of $D$ is too small, the cluster graph obtained at the end of the first stage, $G_{f}$, may have diameter strictly greater than $O(\log^{2+4/\epsilon'} n)$, in which case $\BFST$ may not cover the whole graph $G_{f}$. As a result, once $\SPT$ is constructed from $\BFST$ using Procedure~$\Transform$, some nodes may exist outside the spanning tree $\SPT$. This condition can be detected by the leaves of $\SPT$ and a simple convergecast can be used to check if this condition holds true. In case it does, the root of $\SPT$ can initiate a broadcast over the entire original graph to update the guess of $D$ and run the algorithm with this updated guess. (Note that if the estimate of $D$ is too small, it may still happen that $\BFST$ covers the whole graph $G_{f}$, in which case we correctly compute a low diameter spanning tree $\SPT$ of $G$ and the algorithm terminates.)

This modification increases the time complexity of the algorithm by at most a constant factor, and its message complexity by a factor of at most $O(\log D)$.

\section{The Asynchronous MST Algorithm}
\label{sec:mst}
In this section, we develop a randomized algorithm to construct an MST with high probability for a given graph in $\Tilde{O}(D^{1+\epsilon} + \sqrt{n})$ time with high probability and $\Tilde{O}(m)$ messages with high probability (for any constant $\epsilon > 0$). 

\subsection{High-level Overview of the Algorithm}
\label{sec:high-level}
We implement on an asynchronous network 
a variant of the singularly near optimal {\em synchronous} MST algorithms of~\cite{Elkin17,PanduranganRS17}. 
The algorithm can be divided into three stages. In stage I, we pre-process
the network so that subsequent processes 
are fast and message efficient. Stages II and III correspond to the actual MST algorithm.

 In order to ensure that nodes participate in this multi-stage algorithm in the proper sequence, we append a constant number of bits to each message to indicate the stage number that message corresponds to. A node $u$ knows which stage number it is currently in and can queue received messages that belong to a later stage. These messages will be processed later, once $u$ reaches to the corresponding stage. 

\textbf{Stage I: Pre-Processing the Graph.}
In this stage, we run a few preparatory procedures on the graph.
Specifically, we 
first elect a leader, 
then construct a low diameter spanning tree $\BFSTree$, and finally estimate the diameter of 
$\BFSTree$. 
In more detail, for the first stage we utilize the singularly (near) optimal algorithm of~\cite{KMPP21} to elect a unique leader $\GraphLeader$ in $O(D + \log^2 n)$ time and $O(m \log^2 n)$ messages. 
Subsequently, we run the $\SpanningTreeConstruction(\epsilon)$ algorithm of Section \ref{sec:lowDiamSpanningTree} (for a constant parameter $1 \geq \epsilon > 0$) to construct a low diameter spanning tree $\BFSTree$ on $G$ rooted at $\GraphLeader$. Then, we use a known application of the Wave\&Echo technique (see, e.g., \cite{segall1983distributed,Tel}) to have the root calculate the diameter of the constructed spanning tree $D'$, which we know is an $\tilde{O}(D^\epsilon)$ approximation of the diameter $D$ of the original graph $G$, in $O(D')$ time and $O(n)$ messages. 
Finally, all nodes in the tree participate in a simple broadcast on the spanning tree $\BFSTree$ to send this knowledge of $D'$ to all nodes in the graph in $O(D')$ time and $O(n)$ messages. 

\textbf{Stage II: $\ControlledGHS$.}
The 
$\ControlledGHS$ algorithm, introduced in \cite{garay-sublinear,kutten-domset}, 
is a \emph{synchronous} version of the classical Gallager-Humblet-Spira (GHS) algorithm~\cite{GallagerHS83,peleg-locality} with some modifications, aiming to balance the size and diameter of the resulting fragments.
Here, we convert to the asynchronous setting a variant of the (synchronous) $\ControlledGHS$ as described in \cite{PanduranganRS17,pandurangan2019time}.

Recall that the synchronous GHS algorithm
(see, e.g., \cite{peleg-locality})
consists of $O(\log n)$ phases. In the initial phase, each node is an {\em MST fragment}, by which we mean a connected subgraph of the MST.
In each subsequent phase, every MST fragment finds a minimum-weight outgoing edge (MOE)---these edges are guaranteed to be in the MST
\cite{tarjan}. The MST fragments are merged via the MOEs to form larger fragments.
The number of phases is $O(\log n)$, since the number of MST fragments gets at least halved in each phase.
The message complexity is $O(m + n \log n)$, which is essentially optimal,
and the time complexity is $O(n \log n)$. Unfortunately, the time complexity of the GHS algorithm is not optimal, because much of the communication during a phase uses {\em only the MST fragment edges}, and the diameter of an MST fragment can be significantly larger than the graph diameter $D$ (possibly as large as $\Omega(n)$). 

In order to obtain a time-optimal algorithm,
the $\ControlledGHS$ algorithm controls the growth of the diameter of the MST fragments during merging.
This is achieved by computing, in each phase, a maximal matching
on the fragment forest with additional edges being carefully chosen to ensure enough fragments merge together, and merging fragments accordingly. 
Each phase essentially reduces the number of fragments
by a factor of two, while not increasing the  diameter of any fragment by more than a factor of two. Since the number of phases
of $\ControlledGHS$ is capped at $\max \lbrace \lceil \log_2 \sqrt{n} \rceil, \lceil \log_2 D' \rceil \rbrace$, it produces 
at most $\min \lbrace\sqrt{n}, n/D' \rbrace$ fragments, each of which has diameter $O(D' +  \sqrt{n})$.  These are called {\em base fragments}.
$\ControlledGHS$ up to phase $\max \lbrace \lceil \log_2 \sqrt{n} \rceil, \lceil \log_2 D' \rceil \rbrace$ can be implemented using $\tilde{O}(m)$ messages in
$\tilde{O}(D' + \sqrt{n} )$ rounds in a synchronous network.

Stage II executes
the $\ControlledGHS$ algorithm in an asynchronous network.
We postpone the discussion of the technical details involved in efficiently implementing the asynchronous algorithm to Section~\ref{sec:mst-detail}. 
The main challenge, however, is that the synchronous version heavily relies on the phases being synchronized.
Here, we cannot naively use a synchronizer (such as $\alpha$) for synchronization, as it would have increased
the message complexity substantially. Instead we use a light-weight synchronization that incurs only $\tilde{O}(m)$ overhead
in messages. 

Finally, we ensure that all nodes know the exact number of fragments that were constructed at the end of this phase. The root of each fragment $T$ calculates the number of nodes present in $T$ and forms a tuple consisting of this value and the ID of $T$. Subsequently, each fragment root participates in the upcast of its tuple in the low diameter spanning tree $\BFSTree$ on $G'$. All tuples are accumulated at $\GraphLeader$ in $O(\min \lbrace\sqrt{n}, n/D' \rbrace + D')$ time and $O(n)$ messages. $\GraphLeader$ continues to listen for messages until the total number of nodes in all fragments it has heard from is equal to $n$, i.e., all fragments have been heard from. 
 Now $\GraphLeader$ broadcasts the number of fragments 
over $\BFSTree$ 
 to all nodes in the graph in $O(D')$ time and $O(n)$ messages.

\textbf{Stage III: Merging the Remaining Fragments.}
This stage 
completes the fragment merging process.
However, the merging is done in a ``soft'' manner. The at most $\min \lbrace\sqrt{n}, n/D' \rbrace$ base fragments (constructed at the end of Stage II) are still retained, but each base fragments takes on an additional ID--a cluster ID, initially set to the base fragment ID. (A cluster is a collection of base fragments; at the beginning of this stage, each base fragment forms its own cluster.) Each base fragment finds an MOE to a different cluster, if such an MOE exists, and merging consists of base fragments modifying their associated cluster IDs and marking the corresponding MOE connecting clusters. 
All nodes participate in a simple upcast over $\BFSTree$, where the root of each base fragment is responsible to send up a tuple consisting of its fragment \& cluster IDs, a possible MOE and the associated fragment \& cluster IDs the MOE leads to.\footnote{It is required that each base fragment's root sends up this tuple even if it does not have an MOE (in which case the tuple only has info on the fragment ID and cluster ID of the base fragment). This is to ensure that the nodes detect termination as the root of $\BFSTree$, $\GraphLeader$, already knows the fragment and cluster IDs of the base fragments so it knows how many such messages to wait for.}
It is similar to the approach of \cite{Elkin17,PanduranganRS17}, which uses a BFS tree to upcast these values to the root of tree; here, instead of BFS, we use the
low-diameter spanning tree of Section \ref{sec:lowDiamSpanningTree}. Subsequently, the root calculates the appropriate MOEs (and the fragments they connect and the clusters they lead to) for each cluster and downcast these values. 
Each fragment then performs a broadcast of its (possibly new) cluster ID over the fragment tree (to all nodes within the fragment). This process is repeated for $O(\log n)$ phases until only one cluster remains, which represents the MST of the original graph.

Let us examine each phase $i$ in more detail. Each base fragment finds its respective MOE, if any, and sends it to $\GraphLeader$ via an upcast.\footnote{Note that as the algorithm progresses, two adjacent base fragments may belong to the same overall cluster, possibly resulting in one 
of those base fragments having no MOE to a different cluster.} All fragment leaders can find their MOEs in $O(D' + \sqrt{n})$ time and $O(m)$ messages. Upcasting these values to $\GraphLeader$ using tree $\BFSTree$ takes $O(\min \lbrace\sqrt{n}, n/D' \rbrace + D')$ time and $O(n)$ messages. $\GraphLeader$ locally computes the overall MOEs of the (soft-merged) base fragments
 and then merges them (locally). 
 Subsequently, all nodes of $\BFSTree$ participate in a downcast of these MOEs and modified cluster IDs (that $\GraphLeader$ previously calculated) in $O(D' + \sqrt{n})$ time and $O(n)$ messages. Each base fragment performs a broadcast of its (possibly new) cluster ID to all nodes in its base fragment utilizing the base fragment tree. For all base fragments to do this, it takes a total of $O(D' + \sqrt{n})$ time and $O(n)$ messages.

\subsection{Detailed Algorithm Description}
\label{sec:mst-detail}
We now look at each stage in more detail.

\textbf{Stage I.} In this stage, the nodes first run Procedure~$\LeaderElection$ on $G$ to elect a unique leader $\GraphLeader$ with high probability. As a side benefit, the procedure also wakes up all nodes. Next, the nodes participate in Procedure~$\SpanningTreeConstruction(\epsilon)$ to construct an $\tilde{O}(D^{1+\epsilon})$ diameter spanning tree $\BFSTree$ of $G$ with $\GraphLeader$ as its root. Subsequently, all nodes participate in Procedure~$\DiameterCalculation$ so that $\GraphLeader$ is now aware of the diameter $D'$ of $\BFSTree$. Finally, all nodes participate in $\FragmentBroadcast$ over $\BFSTree$ to transmit this information of $D'$ to all nodes in the graph. (Procedures~$\LeaderElection$, $\DiameterCalculation$ and $\FragmentBroadcast$ are described in Section \ref{sec:toolbox}.)

\textbf{Stage II.} In this stage, the nodes execute 
an asynchronous version of the 
$\ControlledGHS$ algorithm \cite{garay-sublinear,PanduranganRS17,pandurangan2019time}.  
Let us first recall the original (synchronous) $\ControlledGHS$ algorithm. This algorithm merges fragments (subtrees of the MST) in phases, similarly to GHS. However, it guarantees two additional properties to hold at the end of each phase $i$: (a) there are at most $n/2^i$ fragments, and (b) 
each fragment has diameter $O(2^i)$. These guarantees are ensured through two measures. First, at the beginning of phase $i$, only fragments with diameter $\leq 2^i$ will participate in this phase and find MOEs. Second, in a phase $i$, consider the {\em fragment graph} whose ``nodes'' are the fragments (including those that do not participate) and whose edges are all the MOEs found. The algorithm first performs a maximal matching on this fragment graph and removes from the fragment graph edges that do not participate in this matching. 
Then, those fragments who participate in this phase and remain unmatched add their MOEs back to the fragment graph. Connected components of fragments in this final fragment graph then merge together. The algorithm is run from phase $i = 0$ to phase $i = \max \lbrace \lceil \log_2 \sqrt{n} \rceil, \lceil \log_2 D' \rceil \rbrace$.

Let us now explain how to adapt the $\ControlledGHS$ algorithm to the asynchronous setting. First, in order to keep track of the current phase number, all nodes utilize a $\beta$-synchronizer over the tree $\BFSTree$ rooted at $\GraphLeader$.\footnote{Note that we do not use $\beta$ synchronizer to synchronize the beginning of each {\em round}, since this would have been too costly in messages. Using it to synchronize the beginning of each {\em phase} carries a cost we can afford.} 
We next describe how each phase of the algorithm is performed asynchronously. Note that in each phase, each node $u$, belonging to some fragment $F$, maintains information about $F$'s fragment identity $\ID_F$, $F$'s cluster identity\footnote{In stage II, $F$'s cluster ID is its fragment ID. We maintain both values throughout stage II to ensure that the procedures that are called run correctly.} $\ClusterID_F$, $u$'s children in $F$, and $u$'s parent in $F$. 
Initially, each node acts as its own fragment and sets both its fragment ID and cluster ID to its node ID. Each phase consists of 
a constant number of \emph{steps} described below.
We utilize a $\beta$-synchronizer, run by all the nodes over $\BFSTree$, also to keep track of the step number within a given phase via a \emph{step counter}. In phase $i$:

\textbf{Step 1: Each fragment determines if it can participate in the current phase.} Only fragments of diameter $\leq 2^i$, called \emph{active fragments}, may participate in 
phase $i$, so each fragment first checks its diameter, by running 
Procedure~$\DiameterCalculation$, and 
its root $R$ 
determines whether or not $F$ is an active fragment this phase, 
and informs all the other nodes of $F$ using
Procedure~$\FragmentBroadcast$.

\textbf{Step 2: Each active fragment finds its MOE.} Only nodes in active fragments perform the following set of procedures.\footnote{However, nodes that are not in active fragments still reply to messages from their active neighbors. For example, if an inactive node receives a query about its fragment ID, it will respond appropriately.} The nodes of fragment $F$ run Procedure~$\FindMOE$. (Procedure~$\FindMOE$ is described in Section \ref{sec:toolbox}.) Let $\MOEValue_F$ be the resulting MOE $R$ discovers. Subsequently, all nodes run Procedure~$\FragmentBroadcast$ for $R$ to transmit this value to all nodes in $F$.

\textbf{Step 3: Active fragments inform neighboring fragments about MOEs to them.} For every active fragment $F$, each node in $F$ transmits the value of $\MOEValue_F$, if any, to all neighbors. 

\textbf{Step 4: Pre-process the fragment graph before coloring.} 
Consider the directed supergraph $\cH$ formed by all fragments as super nodes 
and MOEs as edges\footnote{Note that there may exist an MOE from an active fragment to an inactive fragment.}. There may exist multiple connected subgraphs within this supergraph. In the next step, the algorithm colors the super nodes in each such subgraph $H^-$. Towards that,  step 4 
preprocesses $H^-$ to form a tree spanning it. 
Note that in each such subgraph $H^-$, there exist (exactly) two fragments with MOEs to each other, and in fact, the two MOEs correspond to the same edge (the two fragments form what is called a core in \cite{GallagerHS83}). 
The fragment with smaller ID among these two 
becomes the root of the resulting 
tree $T(H^-)$ spanning $H^-$ in the supergraph $\cH$, and the other fragment its child. As for all other fragments in $H^-$, they become the child of the other fragment endpoint of their MOE.  

In more detail, consider one such subgraph $H^-$ and a fragment $F$ within it. Recall that the previous step allows the nodes of $F$ to learn whether they share an MOE with another fragment $F'$. Then, a node $u$ in $F$ sets the flag $\IsFragmentRoot$ to $FALSE$ unless one of its incident edges $(u,v)$ is an MOE shared by $F$ as well as another fragment $F'$ and the ID of $F$ is smaller than that of $F'$, in which case $u$ sets the flag $\IsFragmentRoot$ to $TRUE$. 
All nodes in fragment $F$ run Procedure~$\Upcast$ to send this flag to the root of $F$. Since $F$ knows to expect exactly one message, the requirement for the procedure to have termination detection is satisfied and thus $F$ detects the termination of the step. (Procedure~$\Upcast$ is described in Section \ref{sec:toolbox}.) 

\textbf{Step 5: Color fragments.} 
In this step, the trees constructed in the previous step are colored. 
The previous step constructs a forest of rooted trees $T(H^-)$ (of super nodes) that spans the supergraph $\cH$,  
such that whenever super node $u$ has an MOE to super node $v$, $u$ is the child of $v$ in the tree containing them.

Consider the well-known algorithm of Cole and Vishkin~\cite{CV86}, hereafter referred to as Procedure~$\ColeVishkin$. Recall that the algorithm allows one to obtain a $6$-coloring of the tree in $O(\log^* n)$ rounds in the synchronous setting (see e.g., \cite{peleg-locality}). 
We simulate Procedure $\ColeVishkin$ on $\cH$ in an asynchronous setting. To do so, nodes in $G$ to keep track of the round numbers using a $\beta$-synchronizer on $\BFSTree$.
Furthermore, each round of the algorithm is divided into three sub-steps and a $\beta$-synchronizer on $\BFSTree$ is used to keep track of the step numbers. In sub-step one, the root of each fragment $F$ performs any local computation needed, (possibly) resulting in a message $M$ that needs to be transmitted to $F$'s children in $T(H^-)$ eventually. Before that, $M$ is broadcast to all nodes in $F$ via Procedure~$\FragmentBroadcast$.  In sub-step two, all nodes in $F$ transmit $M$ along any incident ``incoming'' MOEs (directed towards $F$ and thus connecting $F$ with its children in $T(H^-)$) and listen for any incoming message $M_i$ transmitted through an ``outgoing'' MOE (directed away from $F$ and thus connecting $F$ with its parent in $T(H^-)$).
In sub-step three, the nodes of $F$ run Procedure~$\Upcast$ to send the received message $M_i$ (if none were received, send a blank message) to the root of $F$.

\textbf{Step 6: Run maximal matching on the colored fragments.} 
At the end of the previous step, we have computed a $6$-coloring of the supergraph $\cH$. Given this coloring, it is straightforward to compute a matching on the supergraph $\cH$ efficiently, even in the asynchronous setting. To do so, fragments (or super nodes) simulate a naive greedy (6-round) synchronous algorithm on $\cH$; in round $i$, super nodes with color $i$ choose an arbitrary, unmatched child super node to match with and informs them of this. (For example, this can be the child fragment with the smallest ID.) This algorithm is simulated in the same way the Cole-Vishkin algorithm is in the previous step. 

\textbf{Step 7: Form the final graph of fragments to be merged into one another.} At the end of the previous step, we obtained a matching on the supergraph $\cH$. If a node (fragment) had diameter $d_{F,i} \leq 2^i$, but did not get matched in the previous step, it adds its MOE as an edge to $\cH$. 

Each node $u$ in the original graph participates in a single transmission to each of its neighbors $v$ to inform $v$ if $(u,v)$ has been re-added or not.

\textbf{Step 8: Merge fragments.} We now finally merge each connected subgraph of $\cH$, obtained at the end of the previous step, into a single fragment. Each newly created fragment takes on the smallest fragment ID from the fragments that merged together to create it. It is easy to see that a combination of a constant number of calls to Procedure~$\FragmentBroadcast$ and Procedure~$\Upcast$ (and a constant number of message on MOEs) results in all nodes in any connected subgraph learning about the minimum fragment ID. The node in the original graph $G$ with this minimum fragment ID becomes the root of the new merged fragment. After which, similarly through a constant number of calls to Procedure~$\FragmentBroadcast$ and Procedure~$\Upcast$ (and a constant number of message on MOEs), we can re-orient the edges so that this newly formed fragment is a tree.

After completing the last phase of the above process, we are almost ready to move to stage III of the algorithm.\footnote{As we use a $\beta$-synchronizer to keep track of which phase a node is in, it is possible to know when $\max \lbrace \lceil \log_2 \sqrt{n} \rceil, \lceil \log_2 D' \rceil \rbrace$ phases are over.} Some final cleanup is first needed. We need two things in order to ensure our subsequent upcasts and downcasts over $\BFSTree$ have termination detection: (i) $\GraphLeader$ needs to be made aware of how many base fragments are present and their IDs and (ii) each node in $\BFSTree$ needs routing information related to any fragment roots located in the subtree rooted at that node in $\BFSTree$.\footnote{Consider a node $u$ and let node $v$ be the root of a fragment located in the subtree rooted at $u$ in $\BFSTree$. We say node $u$ has routing information on $v$ when $u$ knows which of its children in $\BFSTree$ to send a message destined for $v$}

We need each fragment $F$ to inform $\GraphLeader$ of its existence and fragment ID. Now, the root of each fragment $F$, with ID $\ID_{F}$, initiates $\TreeCounting$ to determine the number of nodes in the fragment, $size_F$. (Procedure~$\TreeCounting$ is described in Section \ref{sec:toolbox}.) Subsequently, all nodes in the graph participate in Procedure~$\Upcast$ over $\BFSTree$ where each base fragment's root sends up the tuple $\langle ID_{F}, size_F \rangle$.\footnote{It is important to note that during Procedure~$\Upcast$, each node $u$ in $\BFSTree$ learns about which of its children in $\BFSTree$ lead to which fragment roots. In other words, $u$ learns routing information related to any fragments roots located in the subtree in $\BFSTree$ rooted at $u$, satisfying our second requirement from the previous paragraph.}  
$\GraphLeader$ accumulates these messages until $\sum_F size_F = n$, at which point $\GraphLeader$ knows the exact number of base fragments, say $\NumOfBaseFragments$, and their IDs. Once $\GraphLeader$ recognizes that it has received all the messages, it initiates a broadcast of $\NumOfBaseFragments$ over $\BFSTree$. Now all nodes are aware of the number of base fragments.

\textbf{Stage III.} In this stage, each node $u$ maintains two sets of variables.
One set of variables relates to the base fragment $B$ node $u$ it belongs to at the end of phase two. These variables store information about the base fragment such as the base fragment ID $\ID_B$, $u$'s parents in $B$, and $u$'s children in $B$. The second set of variables relates to what we term a \textit{cluster}, a connected subgraph in $\cH$ consisting of base fragments and MOEs between them, and they store information that includes a cluster ID and cluster edges.  
Each node belonging to base fragment $B$ initially sets its cluster ID $\ClusterID_B$ to be the same as its base fragment ID. Each node $u$ also stores a set of cluster edges adjacent to it in the set $\ClusterEdges_u$, which is initially empty. Edges are added to $\ClusterEdges_u$ in the course of stage III. At the end of stage III, for a given node $u$, the set of edges in the MST is the union of the set of edges in $\ClusterEdges_u$ and its children and parent in $B$. 
Node $\GraphLeader$ maintains, in addition, information on the supergraph $\cH$ formed by the base fragments (including the updated cluster IDs of those base fragments) and any MOE edges that $\GraphLeader$ computes in the phases of stage III, to be described below.

In stage III, each node participates in the following process for $\lceil \log_2 n \rceil$ phases until it terminates. Once again, nodes 
use a $\beta$-synchronizer over $\BFSTree$ 
to keep track of the phase number in stage III. 
In each phase, each base fragment $B$ with root $R_B$, fragment ID $\ID_B$, and cluster ID $\ClusterID_B$ runs Procedure~$\FindMOE$ to find its minimum outgoing edge, say $\MOEValue_B$, to a node with a different cluster ID, if there is any. 
All nodes in the graph then participate in Procedure~$\Upcast$ over $\BFSTree$ to send informatino on the fragments up to $\GraphLeader$. Specifically, each base fragment $B$'s root sends up the tuple consisting of information on $B$ as well as the computed MOE, if any. 

Once $\GraphLeader$ receives this tuple from all base fragments, it locally computes the MOE edges for each cluster in the supergraph $\cH$. Recall that a cluster is a connected subgraph of base fragments in $\cH$. Thus, the MOE from a cluster is really an MOE from one of the base fragments that constitutes it. 
Define $\FinalMOEValue_B$ as the MOE, if any, for base fragment $B$. For each base fragment $B$, $\GraphLeader$ computes its new cluster ID $\ClusterID_B$ (if multiple clusters merge, the smallest cluster ID becomes the ID of the new merged cluster),  
and its $\FinalMOEValue_B$ 
(if the original value of $\FinalMOEValue_B$ broadcast by $B$ was selected as a new edge in $\cH$, $\FinalMOEValue_B$ is set to $\MOEValue_B$, else it is set to a null value).

All nodes participate in Procedure~$\Downcast$ so that $\GraphLeader$ may inform each base fragment's root about its possibly new cluster ID and MOE edge. (Procedure~$\Downcast$ is described in Section \ref{sec:toolbox}.) Subsequently each base fragment participates in Procedure~$\FragmentBroadcast$ to send these values to all nodes in the fragment. Each node updates its cluster ID if needed. If there is information on a new MOE edge out of one of the nodes $u$, then $u$ adds this edge to $\ClusterEdges_u$. 
Once the final phase of stage III is complete, all nodes terminate the algorithm.

\section{Analysis of the MST Algorithm}
\label{sec:analysis}
We argue that Algorithm~$\MSTConstruction$ correctly outputs the MST with high probability and subsequently analyze its running time and message complexity.

It is easy to see that the algorithm faithfully simulates $\ControlledGHS$ in the asynchronous setting. Recall that $\ControlledGHS$ requires us to maintain two properties in each phase of the algorithm:  (i) at the end of phase $i$, there are at most $n/2^i$ fragments and (ii) at the end of phase $i$, each fragment has diameter $O(2^i)$. Since  the algorithm faithfully simulates $\ControlledGHS$, it follows from the analysis of
$\ControlledGHS$ (see e.g., \cite{Elkin17, PanduranganRS17}) that these properties are maintained in stage II. In stage III, they are also maintained via the ``soft merge'' process in a way that is time and message efficient. Note that in stage III, we ensure that those properties hold now on clusters instead of on fragments. These two properties guarantee that after the algorithm is over, there exists one cluster such that all nodes belong to the cluster and the only edges in the cluster are MST edges of the original graph. The high probability guarantee 
comes from the usage of (randomized) Procedures~$\LeaderElection$ and~$\SpanningTreeConstruction$.

We now bound the running time and message complexity in each stage of the algorithm. 
Consider stage I. We initially start with a graph $G$ with $n$ nodes, $m$ edges, and diameter $D$. Stage I involves running one instance of Procedure~$\LeaderElection$, 
one instance of Procedure~$\SpanningTreeConstruction(\epsilon)$ on $G$ to construct a spanning tree $\BFSTree$ of diameter $D' = \tilde{O}(D^{1+\epsilon})$, one instance of Procedure~$\DiameterCalculation$ on $\BFSTree$, and one instance of Procedure~$\FragmentBroadcast$ on $\BFSTree$. From Theorem~\ref{thm:spanningTreeConstruction}, Lemmas~\ref{lem:leader-election} and \ref{lem:frag-broadcast} and Observation~\ref{lem:diameter-calculation}, we have the following lemma. 

\begin{lemma}\label{lem:stage-I-guarantees}
Stage I of Algorithm~$\MSTConstruction$ takes $\Tilde{O}(D^{1+\epsilon})$ time with high probability and $\Tilde{O}(m)$ messages with high probability, for any constant $\epsilon > 0$. 
\end{lemma}

Let us now look at stage II. We utilize a $\beta$-synchronizer over $\BFSTree$ to keep track of the $O(\max \lbrace \lceil \log_2 \sqrt{n} \rceil, \lceil \log_2 D' \rceil \rbrace)$ phases. By Lemma~\ref{lem:beta-synchronizer}, we see that this results in an additive overhead of $O(D')$ time per phase and $O(n)$ messages per phase.

In Step 1, every fragment's nodes participate in one instance of Procedure~$\DiameterCalculation$ in $\Tilde{O}(D + \sqrt{n})$ time and $\Tilde{O}(n)$ messages by Observation~\ref{lem:diameter-calculation}. Subsequently, every fragment's nodes participate in one instance of Procedure~$\FragmentBroadcast$ in $\Tilde{O}(D + \sqrt{n})$ time and $\Tilde{O}(n)$ messages by Lemma~\ref{lem:frag-broadcast}.
In Step 2, nodes participate in Procedure~$\FindMOE$ and then Procedure~$\FragmentBroadcast$, taking a total of $\Tilde{O}(D + \sqrt{n})$ time and $\Tilde{O}(m)$ messages by Lemmas~\ref{lem:finding-moe} and~\ref{lem:frag-broadcast}.
In Step 3, each node sends a message to each of its neighbors in $O(1)$ time and $O(m)$ messages.
In Step 4, every fragment's nodes participate in Procedure~$\Upcast$ on the fragment to send up one message in $\Tilde{O}(D + \sqrt{n})$ time and $\Tilde{O}(D + \sqrt{n})$ messages by Lemma~\ref{lem:upcast}.

In Step 5, nodes simulate Procedure~$\ColeVishkin$ on supergraph $\cH$ using a $\beta$-synchronizer on $\BFSTree$. Recall that Procedure~$\ColeVishkin$ is the Cole-Vishkin~\cite{CV86} algorithm that allows to $6$-color a $n$-node rooted tree in $O(\log^* n)$ rounds and $O(n \log^*n)$ messages. The simulation induces an additive overhead of $O(D')$ time per phase and $O(n)$ messages per sub-step. The correctness of step 5 follows from the correct simulation of each round of Algorithm~$\ColeVishkin$. Finally, each round executes one instance of Procedure~$\FragmentBroadcast$ for the first sub-step, one instance of Procedure~$\Upcast$ for the third sub-step and each node sends at most 1 message to its neighbors in $O(1)$ time and $O(m)$ messages for the second sub-step. Since all of the fragment trees are of depth $\Tilde{O}(D' + \sqrt{n})$ in every phase, each round takes at most $\Tilde{O}(D' + \sqrt{n})$ and $\Tilde{O}(m)$ messages.
Hence, we see that Step 5 results in $\Tilde{O}(D' + \sqrt{n})$ time and $\Tilde{O}(m)$ messages.


In Step 6, in a similar manner to step 5, nodes simulate a synchronous algorithm to compute maximal matching on supergraph $\cH$. The running time of this maximal matching algorithm is subsumed by the running time of Procedure~$\ColeVishkin$. Hence, Step 6 results in $\Tilde{O}(D' + \sqrt{n})$ time and $\Tilde{O}(m)$ messages.
In Step 7, each node in the original graph sends exactly one message to each of its neighbors in $O(1)$ time and $O(m)$ messages. 
In Step 8, nodes participate in a $O(1)$ calls to Procedure~$\FragmentBroadcast$ and Procedure~$\Upcast$. Additionally, each node sends $O(1)$ messages to each of its neighbors. Thus, Step 8 results in $\Tilde{O}(D' + \sqrt{n})$ time and $\Tilde{O}(m)$ messages by Lemmas~\ref{lem:frag-broadcast} and ~\ref{lem:upcast}.

After simulating $\ControlledGHS$ for $O(\max \lbrace \lceil \log_2 \sqrt{n} \rceil, \lceil \log_2 D' \rceil \rbrace)$ phases, we perform some additional procedures before moving on to stage III. Specifically, the nodes of each fragment participate in one instance of Procedure~$\TreeCounting$ in $\Tilde{O}(D' + \sqrt{n})$ time and $O(n)$ messages overall by Observation~\ref{lem:tree-counting}. Subsequently, all nodes in the graph participate in Procedure~$\Upcast$ over $\BFSTree$ to send up information on the  $\min \lbrace\sqrt{n}, n/D' \rbrace$ fragments to $\GraphLeader$ in  $O(\min \lbrace\sqrt{n}, n/D' \rbrace + D')$ time and $O(n)$ messages by Lemma~\ref{lem:upcast}. Finally, $\GraphLeader$ indicates the end of this stage by broadcasting the number of fragments to all nodes in the original graph over $\BFSTree$ in $\Tilde{O}(D')$ time and $\Tilde{O}(n)$ messages.

\begin{lemma}\label{lem:stage-II-guarantees}
Stage II takes $\Tilde{O}(D^{1+\epsilon}+ \sqrt{n})$ time and $\Tilde{O}(m)$ messages. 
\end{lemma}

We now look at stage III. As in stage II, we utilize a $\beta$-synchronizer over $\BFSTree$ to keep track of the $O(\log n)$ phases. By Lemma~\ref{lem:beta-synchronizer}, we see that this results in an additive overhead of $O(D')$ time per phase and $O(n)$ messages per phase. 

Within each phase, the nodes of each base fragment $F$ participate in one instance of  Procedure~$\FindMOE$. From Lemma~\ref{lem:finding-moe}, we see that this takes $O(D' + \sqrt{n})$ time (since the diameter of each base fragment is at most $O(D' + \sqrt{n})$) and $O(m)$ messages since all nodes in the graph participate. Subsequently, all nodes in the original graph participate in Procedure~$\Upcast$ over $\BFSTree$ in order to send up information on $\min \lbrace\sqrt{n}, n/D' \rbrace$ fragments to $\GraphLeader$ in  $O(\min \lbrace\sqrt{n}, n/D' \rbrace + D')$ time and $O(n)$ messages by Lemma~\ref{lem:upcast}. After $\GraphLeader$ performs some computation, all nodes in the graph participate in Procedure~$\Downcast$ over $\BFSTree$ to send down $\min \lbrace\sqrt{n}, n/D' \rbrace$ pieces of information. From Lemma~\ref{lem:downcast}, we see that this takes $O(\min \lbrace\sqrt{n}, n/D' \rbrace + D')$ time and $O(n)$ messages. Finally, nodes from each base fragment participate in Procedure~$\FragmentBroadcast$ resulting in $O(D' + \sqrt{n})$ time and $O(n)$ messages by Lemma~\ref{lem:frag-broadcast}. Thus, we have the following lemma.

\begin{lemma}\label{lem:stage-III-guarantees}
Stage III takes $\Tilde{O}(D^{1+\epsilon}+ \sqrt{n})$ time and $\Tilde{O}(m)$ messages to complete.
\end{lemma}

By Lemmas~\ref{lem:stage-I-guarantees},~\ref{lem:stage-II-guarantees}, and~\ref{lem:stage-III-guarantees} and our initial discussion about correctness, we get the following theorem.


\begin{theorem}\label{the:mst-alg}
Algorithm~$\MSTConstruction$ computes the minimum spanning tree  of an arbitrary graph with high probability in the asynchronous $KT_0$ $\mathcal{CONGEST}$ model   in $\Tilde{O}(D^{1+\epsilon}+ \sqrt{n})$ time with high probability and $\Tilde{O}(m)$ messages with high probability. Furthermore, nodes know their edges in the MST and terminate when the algorithm is over.
\end{theorem}

As a consequence of the above theorem and a theorem due to Mashregi and King~\cite{KMDISC19}[Theorem 1.2] we also get the following result in the $KT_1$ model. 

\begin{theorem}
\label{the:mst-alg-kt1}
There is an asynchronous  algorithm  that computes the minimum spanning tree  of an arbitrary graph with high probability in the asynchronous $KT_1$ $\mathcal{CONGEST}$ model in
$\tilde{O}(D^{1+\epsilon}+ n^{1-2\delta})$ time 
 and $\tilde{O}(n^{3/2+\delta})$ messages  for any small constant $\epsilon > 0$ and for any $\delta \in [0,0.25]$.
\end{theorem}

The above theorem gives the first asynchronous MST algorithm in the  $KT_1$ $\mathcal{CONGEST}$ model that has {\em sublinear time} (for all $D = O(n^{1-\epsilon'})$ for any arbitrarily small
constant $\epsilon' > 0$) and {\em sublinear} messages complexity.

\section{Conclusion and Open Problems}
\label{sec:conclusion}
Recall that while most of the paper deals with the common $KT_0$ model, Theorem \ref{the:mst-alg-kt1} includes a contribution also under the $KT_1$ model. This model has grown in popularity in recent years for two reasons. Firstly, one can claim it is a more natural model \cite{awerbuch1990trade}. Secondly, it allows for the reduction of communication to $o(m)$ messages. Initially, it looked as if this reduction resulted in a significant cost in time complexity, trading off the attempt to go below $\Omega(n)$ time when the diameter is smaller
 \cite{KingKT15}. This went against the direction of the $KT_0$ model, where algorithms managed to be efficient both in time complexity and message complexity \cite{Elkin17,PanduranganRS17,gmyr,GhaffariK18}. 
 Those results, however, were in the synchronous model. Theorem \ref{the:mst-alg-kt1} (together with \cite{KMDISC19}[Theorem 1.2]) is the first result that approaches optimal time while keeping the message complexity at $o(m)$. It would be interesting to see whether this is the best that can be obtained in this direction. Results showing that other tasks can be obtained with $o(m)$ messages but time efficiently in $KT_1$ would also be interesting.

The asynchronous distributed MST algorithm for $KT_0$
 presented here  continues a long line of work in distributed MST algorithms.
Our algorithm  essentially (up to a polylog($n$) factor) matches the respective time and message lower bounds, but for an arbitrarily small constant factor $\epsilon$ in the exponent of $D$ (with respect to time).
Yet, several open problems remain. Is it possible to achieve near singular optimality? That is, can we achieve optimality within a $\polylog(n)$ factor in both time and messages?
This seems related to constructing an $\tilde{O}(D)$ diameter spanning tree   in a singularly optimal fashion which is also open. Our low-diameter spanning tree construction comes close to achieving this, but for a $\tilde{O}(D^{\epsilon})$ factor in the diameter and run time. This is also closely related to constructing a BFS  (or nearly BFS) tree in a singularly optimal fashion.

The tools and techniques used in this paper  for accomplishing various tasks in a (almost) singularly optimal fashion in an asynchronous setting 
can also be useful in solving other fundamental problems such as shortest paths, minimum cut, etc. In particular, the techniques of this paper can be used to
show that the partwise aggregation problem in the low-congestion framework of Ghaffari and Haeupler \cite{low-congestion-mst} can be implemented in the asynchronous setting in  $\tilde{O}(D^{1+\epsilon} + \sqrt{n})$ and $\tilde{O}(m)$ messages. 
Distributed algorithms for fundamental problems such as MST, minimum cut, and shortest paths can be cast as solving
a suitable partwise aggregation problem. In general, due to the $\tilde{\Omega}(D+\sqrt{n})$ lower bound (\cite{peleg-bound,elkin,stoc11}) one cannot hope to solve the partwise aggregation problem in general graphs
faster than $\tilde{O}(D + \sqrt{n})$. Our techniques imply that the following
can be done in a singularly near-optimal fashion in the asynchronous setting: (1) one can reduce the
MST problem into a partwise aggregation problem consisting of 
at most $\min \lbrace\sqrt{n}, n/D^{1+\epsilon} \rbrace$ clusters, each of which has diameter $O(D^{1+\epsilon} +  \sqrt{n})$ and (2) one can construct  a $D^{1+\epsilon}$-diameter spanning tree which can be used as
a low-congestion shortcut.
Both of these  show that the partwise aggregation technique as applicable for the MST yields  $\tilde{O}(D^{1+\epsilon} + \sqrt{n})$ and $\tilde{O}(m)$ messages in asynchronous networks. 
Since the partwise aggregation framework applies to other problems such as  approximate minimum cut
\cite{low-congestion-mst} and approximate single source shortest paths~\cite{hli}, our techniques can imply that these problems 
can also be solved  almost singularly optimally. 

For our singularly optimal algorithms, we focused on being (existentially) optimal in time
with respect to parameters $n$ and $D$ (i.e., with respect to the $\tilde{\Omega}(D+\sqrt{n})$ bound).
An interesting direction of future work is obtaining asynchronous algorithms that 
are ``universally optimal'' (Haeupler, Wajc, and  Zuzic \cite{universal-optimality-mst}) (with respect to time) and also optimal with respect to messages. 

\clearpage
\bibliographystyle{plainurl}
\bibliography{references}

\begin{thebibliography}{10}

\bibitem{AG91}
Yehuda Afek and Eli Gafni.
\newblock Time and message bounds for election in synchronous and asynchronous
  complete networks.
\newblock {\em SICOMP}, 20(2):376--394, 1991.

\bibitem{AM94}
Yehuda Afek and Yossi Matias.
\newblock Elections in anonymous networks.
\newblock {\em Information and Computation}, 113(2):312--330, 1994.

\bibitem{augustine2020latency}
John Augustine, Seth Gilbert, Fabian Kuhn, Peter Robinson, and Suman Sourav.
\newblock Latency, capacity, and distributed minimum spanning tree.
\newblock In {\em 2020 IEEE 40th International Conference on Distributed
  Computing Systems (ICDCS)}, pages 157--167. IEEE, 2020.

\bibitem{awerbuch1985complexity}
Baruch Awerbuch.
\newblock Complexity of network synchronization.
\newblock {\em Journal of the ACM (JACM)}, 32(4):804--823, 1985.

\bibitem{awerbuch-optimal}
Baruch Awerbuch.
\newblock Optimal distributed algorithms for minimum weight spanning tree,
  counting, leader election, and related problems.
\newblock In {\em Proceedings of the 19th ACM Symposium on Theory of Computing
  (STOC)}, pages 230--240, 1987.

\bibitem{A89}
Baruch Awerbuch.
\newblock Distributed shortest paths algorithms (extended abstract).
\newblock In {\em Proceedings of the twenty-first annual ACM symposium on
  Theory of computing}, pages 490--500, 1989.

\bibitem{AGVP90}
Baruch Awerbuch, Oded Goldreich, Ronen Vainish, and David Peleg.
\newblock A trade-off between information and communication in broadcast
  protocols.
\newblock {\em J. ACM}, 37:238--256, 1990.

\bibitem{awerbuch1990trade}
Baruch Awerbuch, Oded Goldreich, Ronen Vainish, and David Peleg.
\newblock A trade-off between information and communication in broadcast
  protocols.
\newblock {\em Journal of the ACM (JACM)}, 37(2):238--256, 1990.

\bibitem{AP90}
Baruch Awerbuch and David Peleg.
\newblock Network synchronization with polylogarithmic overhead.
\newblock In {\em 31st Annual Symposium on Foundations of Computer Science
  (FOCS)}, pages 514--522, 1990.

\bibitem{CDHHLP18}
Yi-Jun Chang, Varsha Dani, Thomas~P. Hayes, Qizheng He, Wenzheng Li, and Seth
  Pettie.
\newblock The energy complexity of broadcast.
\newblock In {\em Proceedings of the 2018 ACM Symposium on Principles of
  Distributed Computing}, PODC '18, page 95–104, New York, NY, USA, 2018.
  Association for Computing Machinery.
\newblock URL: \url{https://doi.org/10.1145/3212734.3212774}, \href
  {http://dx.doi.org/10.1145/3212734.3212774}
  {\path{doi:10.1145/3212734.3212774}}.

\bibitem{CDHP20}
Yi-Jun Chang, Varsha Dani, Thomas~P. Hayes, and Seth Pettie.
\newblock The energy complexity of bfs in radio networks.
\newblock In {\em Proceedings of the 39th Symposium on Principles of
  Distributed Computing}, PODC '20, page 273–282, New York, NY, USA, 2020.
  Association for Computing Machinery.
\newblock URL: \url{https://doi.org/10.1145/3382734.3405713}, \href
  {http://dx.doi.org/10.1145/3382734.3405713}
  {\path{doi:10.1145/3382734.3405713}}.

\bibitem{CV86}
Richard Cole and Uzi Vishkin.
\newblock Deterministic coin tossing with applications to optimal parallel list
  ranking.
\newblock {\em Information and Control}, 70(1):32--53, 1986.

\bibitem{dalal-mst1}
Yogen~K Dalal.
\newblock {\em A Distributed Algorithm for Constructing Minimal Spanning Trees
  in Computer-Communication Networks}.
\newblock Stanford University, 1976.

\bibitem{dalal-mst}
Yogen~K. Dalal.
\newblock A distributed algorithm for constructing minimal spanning trees.
\newblock {\em {IEEE} Trans. Software Eng.}, 13(3):398--405, 1987.

\bibitem{stoc11}
Atish {Das Sarma}, Stephan Holzer, Liah Kor, Amos Korman, Danupon Nanongkai,
  Gopal Pandurangan, David Peleg, and Roger Wattenhofer.
\newblock Distributed verification and hardness of distributed approximation.
\newblock {\em SIAM J. Comput.}, 41(5):1235--1265, 2012.

\bibitem{DP09}
Devdatt~P. Dubhashi and Alessandro Panconesi.
\newblock {\em Concentration of Measure for the Analysis of Randomized
  Algorithms}.
\newblock Cambridge University Press, 2009.
\newblock URL: \url{http://www.cambridge.org/gb/knowledge/isbn/item2327542/}.

\bibitem{elkin-faster}
Michael Elkin.
\newblock A faster distributed protocol for constructing minimum spanning tree.
\newblock {\em Journal of Computer and System Sciences}, 72(8):1282--1308,
  2006.

\bibitem{elkin}
Michael Elkin.
\newblock An unconditional lower bound on the time-approximation trade-off for
  the distributed minimum spanning tree problem.
\newblock {\em {SIAM} J. Comput.}, 36(2):433--456, 2006.

\bibitem{Elkin17}
Michael Elkin.
\newblock A simple deterministic distributed {MST} algorithm, with near-optimal
  time and message complexities.
\newblock In {\em Proceedings of the 2017 ACM Symposium on Principles of
  Distributed Computing (PODC)}, pages 157--163, 2017.

\bibitem{quantum}
Michael Elkin, Hartmut Klauck, Danupon Nanongkai, and Gopal Pandurangan.
\newblock Can quantum communication speed up distributed computation?
\newblock In {\em {ACM} Symposium on Principles of Distributed Computing,
  {PODC}}, pages 166--175. {ACM}, 2014.

\bibitem{faloutsos}
Michalis Faloutsos and Mart Molle.
\newblock A linear-time optimal-message distributed algorithm for minimum
  spanning trees.
\newblock {\em Distributed Computing}, 17(2):151--170, 2004.

\bibitem{fraigniaud2010local}
Pierre Fraigniaud, Amos Korman, and Emmanuelle Lebhar.
\newblock Local mst computation with short advice.
\newblock {\em Theory of Computing Systems}, 47(4):920--933, 2010.

\bibitem{GallagerHS83}
Robert~G. Gallager, Pierre~A. Humblet, and Philip~M. Spira.
\newblock A distributed algorithm for minimum-weight spanning trees.
\newblock {\em ACM Trans. Program. Lang. Syst.}, 5(1):66--77, 1983.

\bibitem{garay-sublinear}
Juan~A. Garay, Shay Kutten, and David Peleg.
\newblock A sublinear time distributed algorithm for minimum-weight spanning
  trees.
\newblock {\em SIAM J. Comput.}, 27(1):302--316, 1998.

\bibitem{low-congestion-mst}
Mohsen Ghaffari and Bernhard Haeupler.
\newblock Distributed algorithms for planar networks {II:} low-congestion
  shortcuts, mst, and min-cut.
\newblock In {\em Proceedings of the Twenty-Seventh Annual {ACM-SIAM} Symposium
  on Discrete Algorithms, {SODA}}, pages 202--219. {SIAM}, 2016.

\bibitem{GhaffariK18}
Mohsen Ghaffari and Fabian Kuhn.
\newblock Distributed {MST} and broadcast with fewer messages, and faster
  gossiping.
\newblock In {\em Proceedings of the 32nd International Symposium on
  Distributed Computing (DISC)}, pages 30:1--30:12, 2018.

\bibitem{gmyr}
Robert Gmyr and Gopal Pandurangan.
\newblock Time-message trade-offs in distributed algorithms.
\newblock In {\em 32nd International Symposium on Distributed Computing, {DISC}
  2018, New Orleans, LA, USA, October 15-19, 2018}, pages 32:1--32:18, 2018.

\bibitem{gupta2003self}
Sandeep~KS Gupta and Pradip~K Srimani.
\newblock Self-stabilizing multicast protocols for ad hoc networks.
\newblock {\em Journal of Parallel and Distributed Computing}, 63(1):87--96,
  2003.

\bibitem{haeupler2018round}
Bernhard Haeupler, D.~Ellis Hershkowitz, and David Wajc.
\newblock Round-and message-optimal distributed graph algorithms.
\newblock In {\em PODC}, pages 119--128, 2018.

\bibitem{hli}
Bernhard Haeupler and Jason Li.
\newblock Faster distributed shortest path approximations via shortcuts.
\newblock {\em arXiv preprint arXiv:1802.03671}, 2018.

\bibitem{HW16}
Bernhard Haeupler and David Wajc.
\newblock A faster distributed radio broadcast primitive: Extended abstract.
\newblock In {\em Proceedings of the 2016 ACM Symposium on Principles of
  Distributed Computing}, PODC '16, page 361–370, New York, NY, USA, 2016.
  Association for Computing Machinery.
\newblock URL: \url{https://doi.org/10.1145/2933057.2933121}, \href
  {http://dx.doi.org/10.1145/2933057.2933121}
  {\path{doi:10.1145/2933057.2933121}}.

\bibitem{universal-optimality-mst}
Bernhard Haeupler, David Wajc, and Goran Zuzic.
\newblock Universally-optimal distributed algorithms for known topologies.
\newblock In {\em {STOC} '21: 53rd Annual {ACM} {SIGACT} Symposium on Theory of
  Computing}, pages 1166--1179. {ACM}, 2021.

\bibitem{KingKT15}
Valerie King, Shay Kutten, and Mikkel Thorup.
\newblock Construction and impromptu repair of an {MST} in a distributed
  network with $o(m)$ communication.
\newblock In {\em Proceedings of the 2015 {ACM} Symposium on Principles of
  Distributed Computing (PODC)}, pages 71--80, 2015.

\bibitem{KorKP11}
Liah Kor, Amos Korman, and David Peleg.
\newblock Tight bounds for distributed {MST} verification.
\newblock In {\em Proc. 28th Symp. on Theoretical Aspects of Computer Science
  ({STACS})}, volume~9 of {\em LIPIcs}, pages 69--80. Schloss Dagstuhl -
  Leibniz-Zentrum f{\"{u}}r Informatik, 2011.

\bibitem{korman2007distributed}
Amos Korman and Shay Kutten.
\newblock Distributed verification of minimum spanning trees.
\newblock {\em Distributed Computing}, 20(4):253--266, 2007.

\bibitem{KormanKP05}
Amos Korman, Shay Kutten, and David Peleg.
\newblock Proof labeling schemes.
\newblock In {\em Proc.24th {ACM} Symp. on Principles of Distributed Computing
  ({PODC})}, pages 9--18, 2005.

\bibitem{KMPP21}
Shay Kutten, William~K. Moses~Jr., Gopal Pandurangan, and David Peleg.
\newblock Singularly near optimal leader election in asynchronous networks.
\newblock In {\em 35th International Symposium on Distributed Computing
  (DISC)}, pages 27:1--27:18, 2021.

\bibitem{jacm15}
Shay Kutten, Gopal Pandurangan, David Peleg, Peter Robinson, and Amitabh
  Trehan.
\newblock On the complexity of universal leader election.
\newblock {\em J. ACM}, 62(1), 2015.

\bibitem{kutten-domset}
Shay Kutten and David Peleg.
\newblock Fast distributed construction of small $k$-dominating sets and
  applications.
\newblock {\em J. Algorithms}, 28(1):40--66, 1998.

\bibitem{LotkerPPP05}
Zvi Lotker, Boaz Patt{-}Shamir, Elan Pavlov, and David Peleg.
\newblock Minimum-weight spanning tree construction in \emph{O}(log log
  \emph{n}) communication rounds.
\newblock {\em {SIAM} J. Comput.}, 35:120--131, 2005.

\bibitem{LotkerPP01}
Zvi Lotker, Boaz Patt{-}Shamir, and David Peleg.
\newblock Distributed {MST} for constant diameter graphs.
\newblock In {\em Proc. 20th {ACM} Symp. on Principles of Distributed Computing
  ({PODC})}, pages 63--71, 2001.

\bibitem{MashreghiK17}
Ali Mashreghi and Valerie King.
\newblock Time-communication trade-offs for minimum spanning tree construction.
\newblock In {\em Proceedings of the 18th International Conference on
  Distributed Computing and Networking (ICDCN)}, 2017.

\bibitem{KMDISC18}
Ali Mashreghi and Valerie King.
\newblock Broadcast and minimum spanning tree with o(m) messages in the
  asynchronous {CONGEST} model.
\newblock In {\em 32nd International Symposium on Distributed Computing, {DISC}
  2018, New Orleans, LA, USA, October 15-19, 2018}, volume 121 of {\em LIPIcs},
  pages 37:1--37:17, 2018.

\bibitem{KMDISC19}
Ali Mashreghi and Valerie King.
\newblock Brief announcement: Faster asynchronous {MST} and low diameter tree
  construction with sublinear communication.
\newblock In Jukka Suomela, editor, {\em 33rd International Symposium on
  Distributed Computing, {DISC} 2019, October 14-18, 2019, Budapest, Hungary},
  volume 146 of {\em LIPIcs}, pages 49:1--49:3, 2019.

\bibitem{MK21}
Ali Mashreghi and Valerie King.
\newblock Broadcast and minimum spanning tree with o(m) messages in the
  asynchronous {CONGEST} model.
\newblock {\em Distributed Computing}, pages 1--17, 2021.

\bibitem{MPX13}
Gary~L. Miller, Richard Peng, and Shen~Chen Xu.
\newblock Parallel graph decompositions using random shifts.
\newblock In {\em Proceedings of the Twenty-Fifth Annual ACM Symposium on
  Parallelism in Algorithms and Architectures}, SPAA '13, page 196–203, New
  York, NY, USA, 2013. Association for Computing Machinery.
\newblock URL: \url{https://doi.org/10.1145/2486159.2486180}, \href
  {http://dx.doi.org/10.1145/2486159.2486180}
  {\path{doi:10.1145/2486159.2486180}}.

\bibitem{PanduranganRS17}
Gopal Pandurangan, Peter Robinson, and Michele Scquizzato.
\newblock A time- and message-optimal distributed algorithm for minimum
  spanning trees.
\newblock In {\em Proceedings of the 49th Annual ACM Symposium on the Theory of
  Computing (STOC)}, pages 743--756, 2017.

\bibitem{PanduranganRS18}
Gopal Pandurangan, Peter Robinson, and Michele Scquizzato.
\newblock The distributed minimum spanning tree problem.
\newblock {\em Bulletin of the {EATCS}}, 125, 2018.

\bibitem{pandurangan2019time}
Gopal Pandurangan, Peter Robinson, and Michele Scquizzato.
\newblock A time- and message-optimal distributed algorithm for minimum
  spanning trees.
\newblock {\em ACM Transactions on Algorithms (TALG)}, 16(1):1--27, 2019.

\bibitem{peleg-locality}
David Peleg.
\newblock {\em Distributed Computing: A Locality Sensitive Approach}.
\newblock SIAM, 2000.

\bibitem{peleg-bound}
David Peleg and Vitaly Rubinovich.
\newblock A near-tight lower bound on the time complexity of distributed
  minimum-weight spanning tree construction.
\newblock {\em {SIAM} J. Comput.}, 30(5):1427--1442, 2000.

\bibitem{rohilla2020efficient}
Deepak Rohilla, Mahendra~Kumar Murmu, and Shashidhar Kulkarni.
\newblock An efficient distributed approach to construct a minimum spanning
  tree in cognitive radio network.
\newblock In {\em First International Conference on Sustainable Technologies
  for Computational Intelligence}, pages 397--407. Springer, 2020.

\bibitem{SS94}
Baruch Schieber and Marc Snir.
\newblock Calling names on nameless networks.
\newblock {\em Information and Computation}, 113(1):80--101, 1994.

\bibitem{segall1983distributed}
Adrian Segall.
\newblock Distributed network protocols.
\newblock {\em IEEE transactions on Information Theory}, 29(1):23--35, 1983.

\bibitem{singh97}
Gurdip Singh.
\newblock Efficient leader election using sense of direction.
\newblock {\em Distributed Computing}, 10(3):159--165, 1997.

\bibitem{spira1977communication}
Philip Spira.
\newblock Communication complexity of distributed minimum spanning tree
  algorithms.
\newblock In {\em Proceedings of the second Berkeley conference on distributed
  data management and computer networks}, 1977.

\bibitem{tarjan}
Robert~Endre Tarjan.
\newblock {\em Data Structures and Network Algorithms}.
\newblock Society for Industrial and Applied Mathematics, 1983.

\bibitem{Tel}
Gerard Tel.
\newblock {\em Introduction to Distributed Algorithms}.
\newblock Cambridge University Press, 1994.

\end{thebibliography}


\end{document}